%% file: 00-main.tex
\def\archive{1}

\ifx\archive\undefined
\documentclass[acmtog]{acmart}
\else
\documentclass[9pt,twocolumn, letterpaper]{article}
\usepackage{amsthm,amsmath,amsfonts}
\usepackage{amssymb}
\usepackage{hyperref}
\usepackage{float}
\usepackage{libertine}
\usepackage{libertinust1math}
\usepackage[T1]{fontenc}
\newcommand{\C}{C}
\usepackage[margin=0.5in]{geometry}
\newtheorem{theorem}{Theorem}
\newtheorem{proposition}{Proposition}
\author{Bolun Wang, Zachary Ferguson, Teseo Schneider, Xin Jiang, Marco Attene, Daniele Panozzo}
\newcommand{\citet}[1]{\cite{#1}}
\fi

\usepackage{booktabs} 
\ifx\archive\undefined
\setcopyright{acmlicensed}
\acmJournal{TOG}
\acmMonth{10}
\acmYear{2021}
\acmVolume{40}
\acmNumber{5}
\acmArticle{188}
\acmPrice{15.00}
\acmDOI{10.1145/3460775}

\fi
\usepackage{siunitx}
\usepackage{algorithm}
\usepackage{xspace}
\usepackage[noend]{algpseudocode}
\usepackage{colortbl}
\usepackage{color,soul}
\input{97-macros}
\input{98-symbols}

\begin{document}

\title{A Large Scale Benchmark and an Inclusion-Based Algorithm for Continuous Collision Detection}

\ifx\archive\undefined
  \author{Bolun Wang}
  \affiliation{%
    \department{LMIB \& NLSDE \& School of Mathematical Sciences \& Shenyuan Honor College}
    \institution{Beihang University}
  }
  \affiliation{%
    \institution{ New York University}
  }
  \email{wangbolun@buaa.edu.cn}
  \authornote{Joint first authors}

  \author{Zachary Ferguson}
  \affiliation{%
    \institution{New York University}
  }
  \email{zfergus@nyu.edu}
  \authornotemark[1]

  \author{Teseo Schneider}
  \affiliation{%
    \institution{New York University}
  }
  \affiliation{%
    \institution{University of Victoria}
  }

  \author{Xin Jiang}
  \affiliation{%
    \department{LMIB \& NLSDE \& School of Mathematical Sciences}
    \institution{ Beihang University}
  }
  \affiliation{%
    \institution{ Peng Cheng Laboratory in Shenzhen}
  }

  \author{Marco Attene}
  \affiliation{%
    \institution{IMATI - CNR}
  }

  \author{Daniele Panozzo}
  \affiliation{%
    \institution{New York University}
  }
  \email{panozzo@nyu.edu}

  \begin{abstract}
    \input{00-abstract}
  \end{abstract}

  %

  %
  %

\fi

\ifx\archive\undefined
  \begin{teaserfigure}
    \input{teaser}
  \end{teaserfigure}
  \maketitle
\else
  \date{}
  \maketitle
  \input{00-abstract}
  \begin{figure*}[ht]\centering
    \input{teaser}
  \end{figure*}
\fi

\input{01-intro.tex}

\input{02-related.tex}

\input{02.1-preliminaries}

\input{03-benchmark}
\input{04-method}

\input{05.1-msccd}

\input{05-application}
\input{06-conclusion}
\input{07-acknowledgement}



\input{bib.tex}
\appendix
\input{10-appendix}

\end{document}

%% file: 97-macros.tex
\definecolor{FNColor}{rgb}{0.83, 0.19, 0.19}
\definecolor{FPColor}{rgb}{0.3, 0.52, 0.89}
\definecolor{ZachColor}{rgb}{0.9, 0.5, 0}
\definecolor{Red}{rgb}{1, 0, 0}
\definecolor{Pink}{rgb}{1, 0, 1}
\definecolor{TeseoColor}{rgb}{0.15, 0.68, 0.38}
\definecolor{BolunColor}{rgb}{0, 0, 0.9}
\definecolor{MColor}{rgb}{0.1, 0.8, 0.7}
\definecolor{TableHeader}{rgb}{0.556,0.231,0.396}
\definecolor{ReviewColor}{rgb}{0.1, 0.8, 0.7}
\definecolor{ReviewwColor}{rgb}{0.4, 0.3, 0.9}
\definecolor{Review3Color}{rgb}{0.8, 0, 0.8}
\newcommand{\review}[1]{{\leavevmode\color{ReviewColor} #1}}
\newcommand{\revieww}[1]{{\leavevmode\color{ReviewwColor} #1}}
\newcommand{\reviewww}[1]{{\leavevmode\color{Review3Color} #1}}

\providecommand{\finalversion}{1} 
\ifthenelse{\equal{\finalversion}{1}}
{
	\renewcommand{\review}[1]{{#1}}
    \renewcommand{\revieww}[1]{{#1}}
    \renewcommand{\reviewww}[1]{{#1}}
}
{}

\definecolor{dpurple}{rgb}{0.4,0,0.4}
\definecolor{dgreen}{rgb}{0,0.3,0}
\definecolor{dgrey}{rgb}{0.5,0.5,0.5}
\definecolor{orange}{rgb}{0.9,0.4,0}

\usepackage{listings}
\usepackage{adjustbox}
\lstdefinelanguage{Simple}{
    morecomment=[l]{//},
    morekeywords={if, for, else, break, continue, return, in, length, not, struct, and, or, NULL},
    morekeywords = [2]{NO_COLLISION, COLLISION},
    morekeywords = [3]{prism_bounding_box, prism_edge_cube, cube_edge_prism
parity_check},
    morekeywords = [4]{xxx}
    otherkeywords = {[,]},
}

%% file: 98-symbols.tex
\usepackage{mathtools}
\usepackage{relsize}

\let\originalleft\left
\let\originalright\right
\renewcommand{\left}{\mathopen{}\mathclose\bgroup\originalleft}
\renewcommand{\right}{\aftergroup\egroup\originalright}

\renewcommand{\geq}{\geqslant}
\renewcommand{\leq}{\leqslant}

\newcommand{\reals}{\mathbb{R}}

\renewcommand{\C}{C}

\newcommand{\p}{p}

\renewcommand{\v}{v}

\renewcommand{\d}{d}

\newcommand{\RR}{\reals}

\newcommand{\incl}{\mathlarger{\mathlarger{\mathlarger{{\square}}}}}

\newcommand{\handcrafted}{handcrafted\xspace} 
\newcommand{\Handcrafted}{Handcrafted\xspace}

\newcommand{\LE}[1]{\mathcal{L}(#1)}
\newcommand{\RI}[1]{\mathcal{R}(#1)}

%% file: 00-abstract.tex
We introduce a large-scale benchmark for continuous collision detection (CCD) algorithms, composed of queries manually constructed to highlight challenging degenerate cases and automatically generated using existing simulators to cover common cases. We use the benchmark to evaluate the accuracy, correctness, and efficiency of state-of-the-art continuous collision detection algorithms, both with and without minimal separation.

We discover that, despite the widespread use of CCD algorithms, existing algorithms are either: (1) correct but impractically slow, (2) efficient but incorrect, introducing false negatives which will lead to interpenetration, or (3) correct but over conservative, reporting a large number of false positives which might lead to inaccuracies when integrated into a simulator.

By combining the seminal interval root-finding algorithm introduced by Snyder in 1992 with modern predicate design techniques, we propose a simple and efficient CCD algorithm. This algorithm is competitive with state-of-the-art methods in terms of runtime while conservatively reporting the time of impact and allowing an explicit trade-off between runtime efficiency and the number of false positives reported.

%% file: teaser.tex
\parbox{.31\linewidth}{
\parbox{0.02\linewidth}{\centering\rotatebox{90}{\scriptsize{Number of False Positives}}}\hfill
\parbox{0.95\linewidth}{\centering\includegraphics[width=\linewidth]{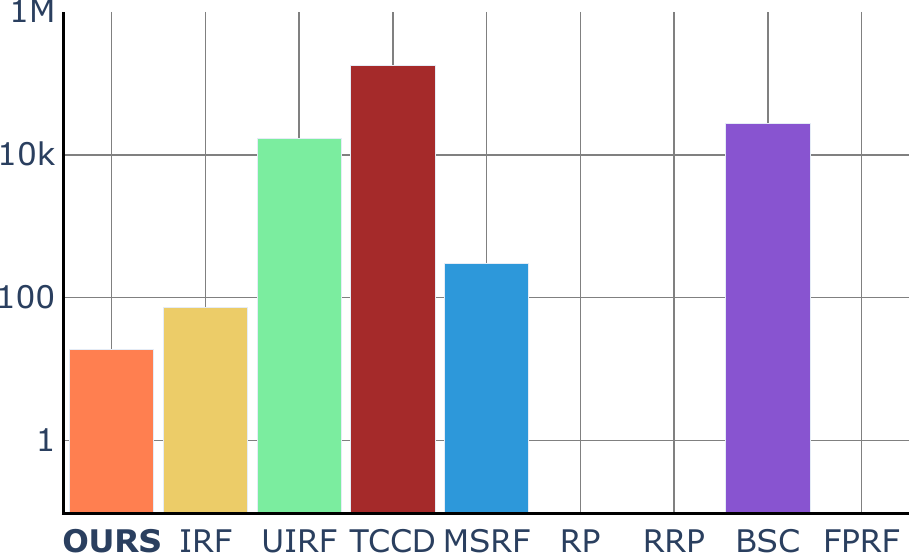}}}\hfill
\parbox{.31\linewidth}{
\parbox{0.02\linewidth}{\centering\rotatebox{90}{\scriptsize{Number of False Negatives}}}\hfill
\parbox{0.95\linewidth}{\centering\includegraphics[width=\linewidth]{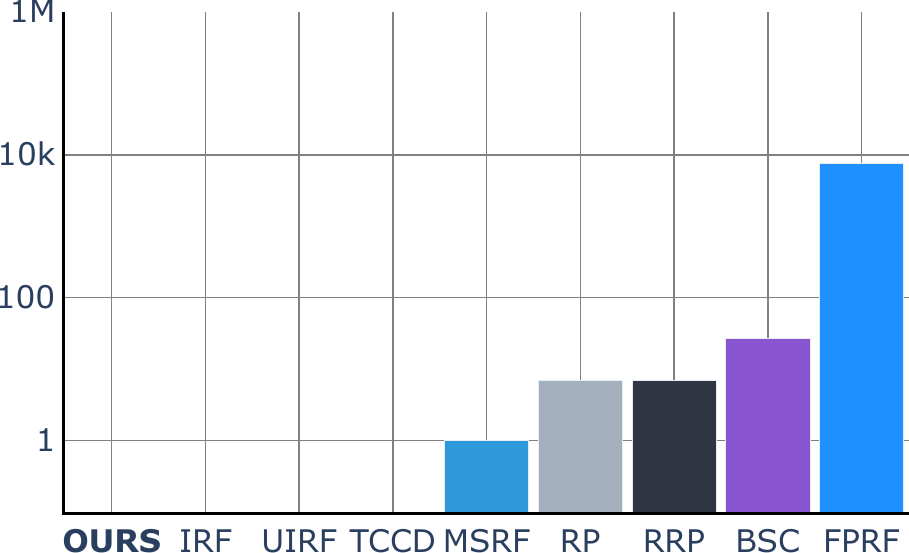}}}\hfill
\parbox{.31\linewidth}{
\parbox{0.02\linewidth}{\centering\rotatebox{90}{\scriptsize{Average Time (\SI[]{}{\micro\second})}}}\hfill
\parbox{0.95\linewidth}{\centering\includegraphics[width=\linewidth]{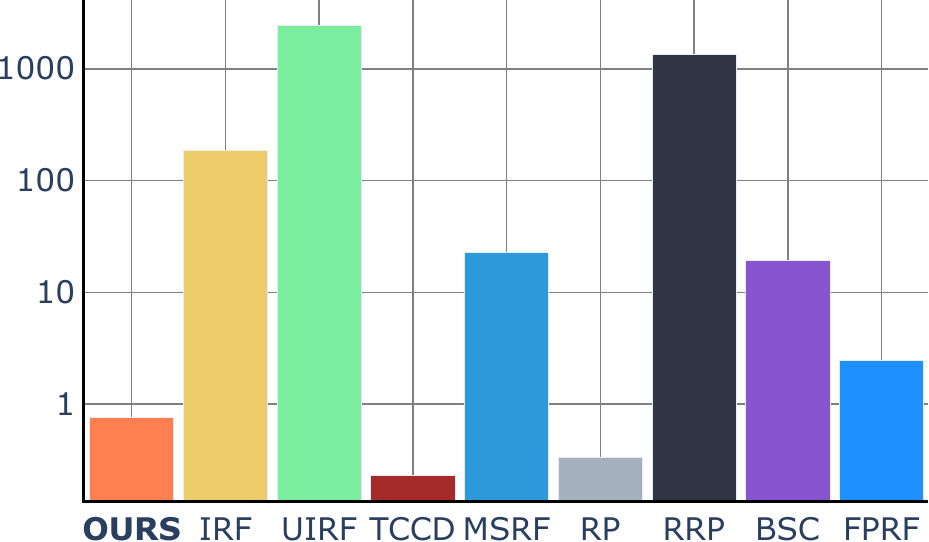}}}\hfill
\caption{An overview of the results of our study of different CCD methods run on 60 million queries (both vertex-face and edge-edge). For each method, we show the number of false positives (i.e., the method detects a collision where there is none), the number of false negatives (i.e., the method misses a collision), and the average run time.  Each plot reports results in a logarithmic scale. False positives and negatives are computed with respect to the ground truth computed using Mathematica~\cite{Mathematica}. Acronyms are defined in Section~\ref{sec:comparison}.}

%% file: 01-intro.tex
\section{Introduction}

\review{Collision detection and response are two separate, yet interconnected, problems  in computer graphics and scientific computing. Collision detection specializes in finding when and if two objects collide, while collision response uses this information to deform the objects following physical laws. A large research effort has been invested in the latter problem,  assuming that collision detection can be solved reliably and efficiently. In this study we focus on the former, using an experimental approach based on large scale testing. We use existing collision response methods to generate collision detection queries to investigate the pros and cons of existing collision detection algorithms.} 

Static collision detection is popular in interactive applications due to its efficiency, its inability to detect collisions between fast moving objects passing through each other (tunneling) hinders its applicability. To address this limitation, continuous collision detection (CCD) methods have been introduced: by solving a more computationally intensive problem, usually involving finding roots of a low-degree polynomial, these algorithms can detect any collision happening in a time step, often assuming linear trajectories. 

The added robustness makes this family of algorithms popular, but they can still fail due to floating-point rounding errors. Floating point failures are of two types: false negatives, i.e., missed collisions, which lead to interpenetration, and false positives, i.e., detecting collisions when there are none. 

\review{Most collision response algorithms can tolerate minor imperfections, using heuristics to recover from physically invalid states (in reality, objects cannot inter-penetrate). However, these heuristics have parameters that needs to be tuned for every scene to ensure stability and faithfulness in the simulation~\cite{Li2020IPC}. Recently, the collision response problem has been reformulated to avoid the use of heuristics, and the corresponding parameter tuning, by disallowing physically invalid configurations~\cite{Li2020IPC}}.
\reviewww{For instance, in the attached video, the method in~\cite{Li2020IPC} cannot recover from interpenetration after the CCD misses a collision leading to an unnatural ``sticking'' and eventual failure of the simulation.}
\review{This comes with a heavier burden on the CCD algorithm used, which should never report false negatives.} 

We introduce a large benchmark of CCD queries with ground truth computed using the \emph{exact, symbolic solver} of Mathematica~\cite{Mathematica}, and evaluate the correctness (lack of false negatives), conservativeness (false positive count), and runtime efficiency of existing state of the art algorithms. The benchmark is composed of both manually designed queries to identify degenerate cases (building upon~\cite{erleben2018methodology}) and a large collection of real-world queries extracted from simulation sequences. 
On the algorithmic side, we select representative algorithms from the three main approaches existing in the literature for CCD root-finding: \emph{inclusion-based bisection methods}~\cite{Snyder1993Interval,Redon2002fast}, \emph{numerical methods}~\cite{vouga2010AVCM,wang2015tightccd}, and \emph{exact methods}~\cite{brochu2012efficient,tang2014fast}.
Thanks to our benchmark, we identified missing cases that were not handled by previous methods, and we did a best effort to fix the corresponding algorithms and implementations to account for these cases.

The surprising conclusion of this study (Section~\ref{sec:comparison}) is that the majority of the existing CCD algorithms produce false negatives, except three: (1) symbolic solution of the system and evaluation with exact arithmetic computed using Mathematica~\cite{Mathematica}, (2) Bernstein sign classification (BSC) with conservative error analysis~\cite{wang2015tightccd}, and (3) inclusion-based bisection root finding~\cite{Snyder1993Interval,Redon2002fast}. (1) is extremely expensive and, while it can be used for generating the ground truth, it is impractical in simulation applications. (2) is efficient but generates many false positives and the number of false positives depends on the geometric configuration and velocities involved. (3) is one of the oldest methods proposed for CCD. It is slow compared to state of the art algorithms, but it is correct and allows precise control of the trade-off between false positives and computational cost. 

This extensive analysis and benchmark inspired us to introduce a specialization of the classical inclusion-based bisection algorithm proposed in~\cite{snyder1992interval} to the specific case of CCD for triangular meshes (Section~\ref{sec:method}). The major changes are: a novel inclusion function, an efficient strategy to perform bisection, and the ability to find CCD roots with minimal separation (Section~\ref{sec:msccd}). Our novel inclusion function:
\begin{enumerate}
\item is tighter leading to smaller boxes on average thus making our method more accurate (i.e., less false positive);
\item reduces the root-finding problem into the iterative evaluation of a Boolean function, which allows replacing explicit interval arithmetic with a more efficient floating point filtering;
\item can be vectorized with \texttt{AVX2} instructions.
\end{enumerate}
With these modifications, our inclusion-based bisection algorithm is only $3\times$ slower on average than the fastest inaccurate CCD algorithm. At the same time it is provably conservative, provides a controllable ratio of false positives (within reasonable numerical limits), supports minimal separation, and reports the time of impact. We also discuss how to integrate minimal separation CCD in algorithms employing a line search to ensure the lack of intersections, which are common in locally injective mesh parametrization and have been recently introduced in physical simulation by~\citet{Li2020IPC}.

\review{Our dataset is available at the  \href{https://archive.nyu.edu/handle/2451/61518}{NYU Faculty Digital Archive}}, while the implementation of all the algorithms compared in the benchmark, a reference implementation of our novel inclusion-based bisection algorithm, and scripts to reproduce all results (Section~\ref{sec:benchmark}) \review{are available on \href{https://continuous-collision-detection.github.io/}{our project web page}.} We believe this dataset will be an important element to support research in efficient and correct CCD algorithms, while our novel inclusion-based bisection algorithm is a practical solution that will allow researchers and practitioners to robustly check for collisions in applications where a $3\times$ slowdown in the CCD (which is usually only one of the expensive steps of a simulation pipeline) will be preferable over the risk of false negatives or the need to tune CCD parameters.

%% file: 02-related.tex
\section{Related Work}\label{sec:related}
We present a brief overview of the previous works on continuous collision detection for triangle meshes. \review{Our work focuses only on CCD for deformable triangle meshes and we thus exclude discussing methods approximating collisions using proxies (e.g., \citet*{Hubbard1995Interactive, Mirtich1996Impulse}).}

\paragraph{Inclusion-Based Root-Finding.} The generic algorithm in the seminal work of~\citet{snyder1992interval} on interval arithmetic for computer graphics is a conservative way to find collisions~\cite{VonHerzen:1990,Snyder1993Interval,Redon2002fast}. This approach uses inclusion functions to certify the existence of roots within a domain, using a bisection partitioning strategy. Surprisingly, this approach is not used in recent algorithms despite being provably conservative and simple. Our algorithm is based on this approach, but with two major extensions to improve its efficiency (Section~\ref{sec:method}).

\paragraph{Numerical Root-Finding.} The majority of CCD research focuses on efficient and accurate ways of computing roots of special cubic polynomials. Among these, a most popular cubic solver approach is introduced by~\citet{provot1997collision}, in which a cubic equation is solved to check for coplanarity, and then the overlapping occurrence is validated to determine whether a collision actually occurs. Refined constructions based on this idea have been introduced for rigid~\cite{Redon2002fast,kim2003collision} and deformable~\cite{Hutter2007optimized,Tang2011VolCCD} bodies.
However, all of these algorithms are based on floating-point arithmetic, requiring numerical thresholds to account for the unavoidable rounding errors in the iterative root-finding procedure. In fact, even if the cubic polynomial is represented exactly, its roots are generally irrational and thus not representable with floating-point numbers. Unfortunately, the numerical thresholds make these algorithms robust only for specific scenarios, and they can in general introduce false negatives.
Our approach has a moderately higher runtime than these algorithms, but it is guaranteed to avoid false negatives without parameter tuning. We benchmark~\citet{provot1997collision} using the implementation of~\citet{vouga2010AVCM} in Section~\ref{sec:benchmark}.

For most applications, false positives are less problematic than false negatives since a false negative will miss a collision, leading to interpenetration and potentially breaking the simulation.  
\citet{Tang2010fast} propose a simple and effective ﬁlter which can reduce both the number of false positives and the elementary tests between the primitives. \citet{Wang2014Defending} and~\citet{wang2015tightccd} improve its reliability by introducing forward error analysis, in which error bounds for floating-point computation are used to eliminate false positives. We benchmark the representative method of~\citet{wang2015tightccd} in Section~\ref{sec:benchmark}.

\paragraph{Exact Root-Finding.}
\citet{brochu2012efficient} and~\citet{tang2014fast} introduce algorithms relying on exact arithmetic to provide exact continuous collision detection. However, after experimenting with their implementations and carefully studying their algorithms, we discovered that they cannot always provide the exact answer (Section~\ref{sec:benchmark}). \citet{brochu2012efficient} rephrase the collision problem as counting the number of intersections between a ray and the boundary of a subset of $\RR^3$ bounded by bilinear faces. The ray casting and polygonal construction can be done using rational numbers (or more efficiently with floating point expansions) to avoid floating-point rounding errors. In~\cite{tang2014fast} the CCD queries are reduced to the evaluation of the signs of Bernstein polynomials and algebraic expressions, using a custom root finding algorithm. Our algorithm uses the geometric formulation proposed in~\cite{brochu2012efficient}, but uses a bisection strategy instead of ray casting to find the roots.
We benchmark both~\cite{brochu2012efficient} and~\cite{tang2014fast} in Section~\ref{sec:benchmark}.

\paragraph{Minimal Separation.}
Minimal separation CCD (MSCCD)~\cite{provot1997collision,Stam2009Nucleus,harmon2011IAGM,Libin:2019} reports collisions when two objects are at a (usually small) user-specified distance. These approaches have two main applications: (1) a minimal separation is useful in fabrication settings to ensure that the fabrication errors will not lead to penetrations, and (2) a minimal separation can ensure that, after floating-point rounding, two objects are still not intersecting, an invariant which must be preserved by certain simulation codes~\cite{harmon2011IAGM,Li2020IPC}. We benchmark~\cite{harmon2011IAGM} in Section~\ref{sec:msccd-results}. Our algorithm supports a novel version of minimal separation, where we use the $L^\infty$ norm instead of $L^2$ (Section~\ref{sec:msccd-method}).

\paragraph{Collision Culling.}
An orthogonal problem is efficient high-level collision culling to quickly filter out primitive pairs that do not collide in a time step. Since in this case it is tolerable to have many false positives, it is easy to find conservative approaches that are guaranteed to not discard potentially intersecting pairs \cite{Curtis2008Fast,Tang2009ICCD,Wong2006randomized,Tang2008Adjacency,volino1994efficient,provot1997collision,Mezger2003Hierarchical,Schvartzman2010Star,Pabst2010Fast,Zheng2012Energy,Zhang2007Continous,Govindaraju2005collision}. Any of these approaches can be used as a preprocessing step to any of the CCD methods considered in this study to improve performance.

\paragraph{Generalized Trajectories.}
The linearization of trajectories commonly used in collision detection is a well-established, practical approximation, ubiquitous in existing codes. There are, however, methods that can directly detect collisions between objects following polynomial trajectories~\cite{Pan2012Collision} or rigid motions~\cite{Tang2009C2A, Redon2002fast, canny1986collision, Zhang2007Continous}, and avoid the approximation errors due to the linearization. Our algorithm currently does not support curved trajectories and we believe this is an important direction for future work.

%% file: 02.1-preliminaries.tex
\section{Preliminaries and Notation}\label{sec:preliminaries}

Assuming that the objects are represented using triangular meshes and that every vertex moves in a linear trajectory in each time step, the first collision between moving triangles can happen either when a vertex hits a triangle, or when an edge hits another edge. 

Thus a continuous collision detection algorithm is a procedure that, given a vertex-face or edge-edge pair, equipped with their \emph{linear trajectories}, determines if and when they will touch. Formally, for the vertex-face CCD, given a vertex $\p$ and a face with vertices $\v_1, \v_2, \v_3$ at two distinct time steps $t^0$ and $t^1$ (we use the superscript notation to denote the time, i.e., $\p^0$ is the position of $\p$ at $t^0$), the goal is to determine if at any point in time between $t^0$ and $t^1$ the vertex is contained in the moving face. Similarly for the edge-edge CCD the algorithm aims to find if there exists a $t\in[t^0, t^1]$ where the two moving edges $(\p_1^t, \p_2^t)$ and $(\p_3^t, \p_4^t)$ intersect. We will briefly overview and discuss the pros and cons of the two major formulations present in the literature to address the CCD problem: \emph{multi-variate} and \emph{univariate}.

\paragraph{Multivariate CCD Formulation}
The most direct way of solving this problem is to parametrize the trajectories with a parameter $t \in [0,1]$ (i.e., $\p_i(t) = (1-t)\p_i^0 + t\p_i^1$ and $\v_i(t) = (1-t)\v_i^0 + t\v_i^1$) and write a multivariate polynomial whose roots correspond to intersections. That is finding the roots of
\[
F_{\text{vf}} \colon \Omega_{\text{vf}} = [0,1] \times \{u,v \geq 0 | u+v \leq 1\} \to \RR^3
\]
with
\begin{equation}\label{eq:F-vf}
F_{\text{vf}}(t,u,v) = \p(t) - \big((1-u-v)\v_1(t) + u\v_2(t) + v\v_3(t)\big),
\end{equation}
for the vertex-face case. Similarly for the edge-edge case the goal is to find the roots of
\[
F_{\text{ee}} \colon \Omega_{\text{ee}} = [0,1] \times [0,1]^2 \to \RR^3
\]
with
\begin{equation}\label{eq:F-ee}
F_{\text{ee}}(t,u,v) = \big((1-u)\p_1(t) + u\p_2(t)\big) - \big((1-v)\p_3(t) + v\p_4(t)\big).
\end{equation}
In other words, the CCD problem reduces to determining if $F$ has a root in $\Omega$ (i.e., there is a combination of valid $t,u,v$ for which the vector between the point and the triangle is zero)~\cite{brochu2012efficient}. 
The main advantage of this formulation is that it is direct and purely algebraic: there are no degenerate or corner cases to handle. The intersection point is parameterized in time and local coordinates and the CCD problem reduces to multivariate root-finding. However, finding roots of a system of quadratic polynomials is difficult and expensive, which led to the introduction of the univariate formulation.

\paragraph{Univariate CCD Formulation}

An alternative way of addressing the CCD problem is to rely on a \emph{geometric} observation: two primitives intersects if the four points (i.e., one vertex and the three triangle's vertices or the two pairs of edge's endpoints) are coplanar~\cite{provot1997collision}. This observation has the major advantage of only depending on time, thus the problem becomes finding roots in a univariate cubic polynomial: 
\begin{equation}\label{eq:univariate}
f(t) = \langle n(t), q(t)\rangle = 0,
\end{equation}
with
\[
n(t) = \big(\v_2(t)-\v_1(t)\big) \times \big(\v_3(t)-\v_1(t)\big)~\text{and}~
q(t) = \p(t)-\v_1(t)
\]
for the vertex-face case and
\[
n(t) = \big(\p_2(t)-\p_1(t)\big) \times \big(\p_4(t)-\p_3(t)\big)~\text{and}~
q(t) = \p_3(t)-\p_1(t)
\]
for the edge-edge case. Once the roots $t^\star$ of $f$ are identified, they need to be filtered, as not all roots correspond to actual collisions. While filtering is straightforward when the roots are finite, special care is needed when there is an infinite number of roots, such as when the two primitives are moving on the same plane. Handling these cases, especially while accounting for floating point rounding, is very challenging.

%% file: 03-benchmark.tex
\section{Benchmark}\label{sec:benchmark}

\subsection{Dataset} 
\begin{figure}
    \centering
    \includegraphics[width=0.9\linewidth]{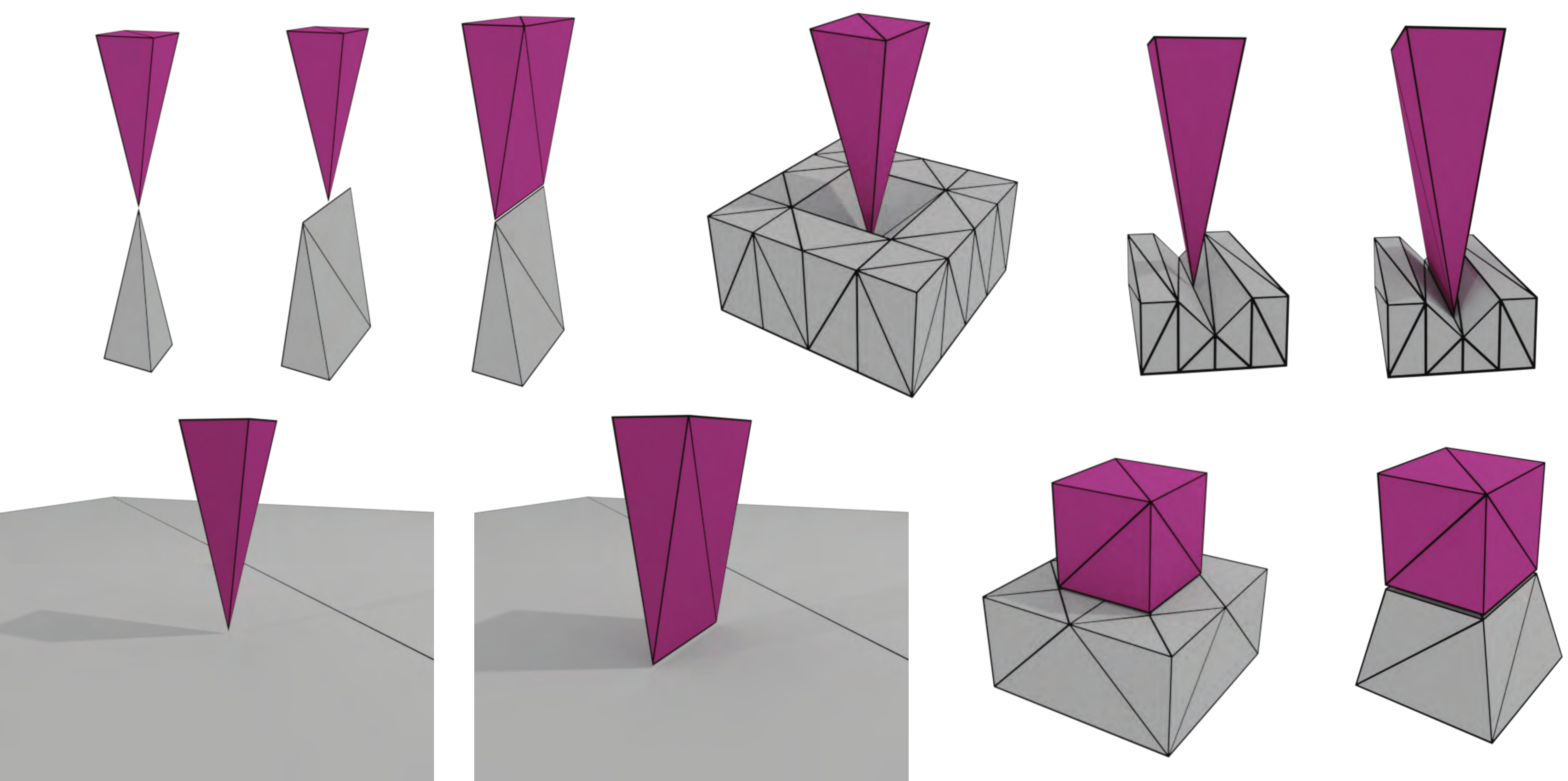}
    \caption{\review{Scenes from~\citet{erleben2018methodology} that are used to generate a large part of the \handcrafted dataset.}}
    \label{fig:hand-dataset}
\end{figure}

\begin{figure}
    \centering
    \includegraphics[width=0.9\linewidth]{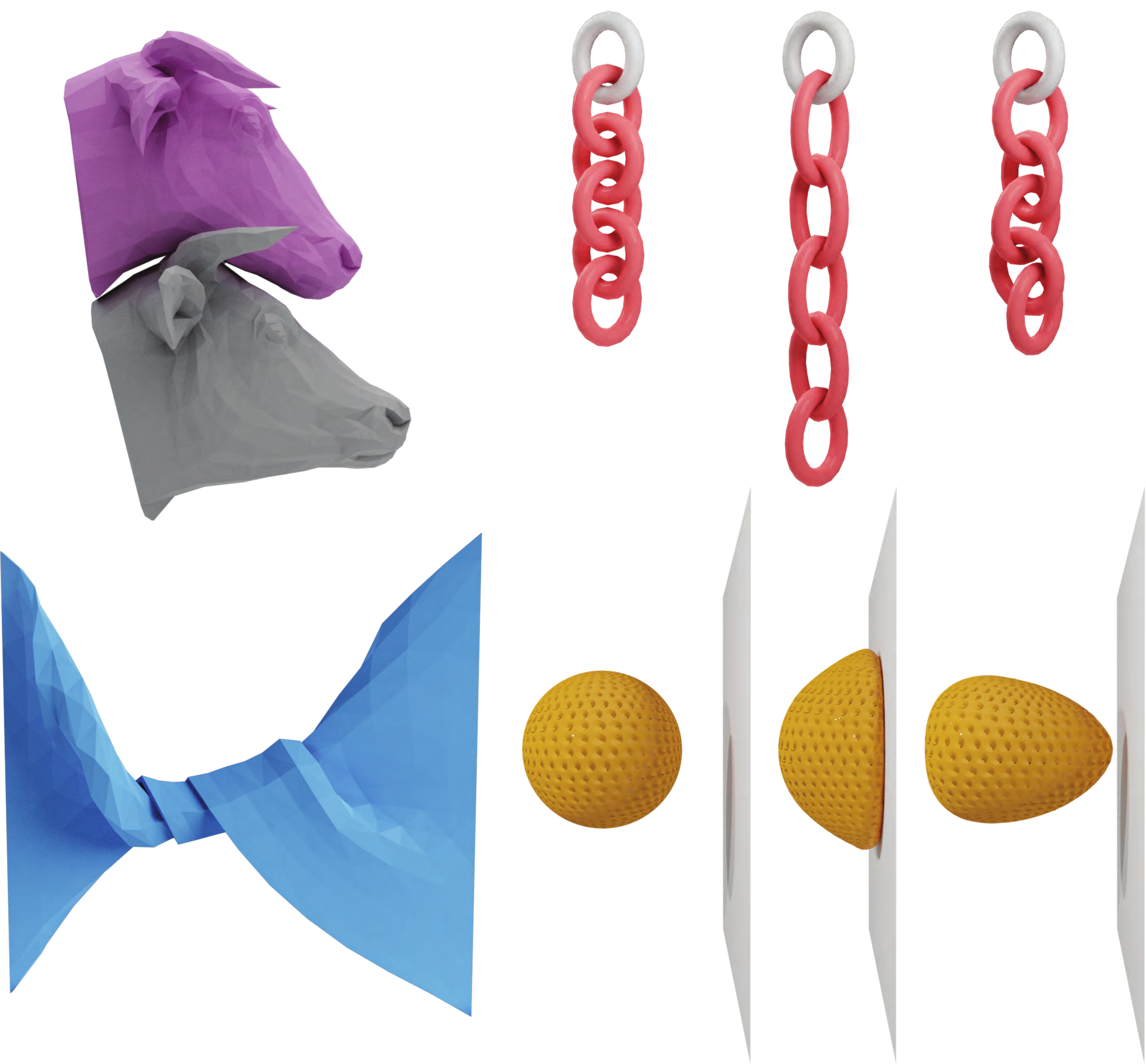}
    \caption{The scenes used to generate the simulation dataset of queries. We use two simulation methods: (top) a sequential quadratic programming (SQP) method with constraints and active set update from~\citet{verschoor2019efficient} and (bottom) the method proposed by~\citet{Li2020IPC}.}
    \label{fig:sim-dataset}
\end{figure}

We crafted two datasets to compare the performance and correctness of CCD algorithms: (1) a \emph{\handcrafted} dataset that contains over 12 thousand point-triangle and 15 thousand edge-edge queries, and (2) a \emph{simulation} dataset that contains over 18 million point-triangle and 41 million edge-edge queries. To foster replicability, we describe the format of the dataset in Appendix~\ref{app:format}.

The \handcrafted queries are the union of queries \review{simulated with~\cite{Li2020IPC} from the scenes in}~\cite{erleben2018methodology}  (Figure~\ref{fig:hand-dataset}) and a set of handcrafted pairs for degenerate geometric configurations. These include: point-point degeneracies, near collisions (within a floating-point epsilon from collision), coplanar vertex-face and edge-edge motion (where the function $f$~\eqref{eq:univariate} has infinite roots), degenerated function $F_{\text{vf}}$ and $F_{\text{ee}}$, and CCD queries with two or three roots.

The simulation queries were generated by running four nonlinear elasticity simulations. The first two simulations (Figure~\ref{fig:sim-dataset} top row) use the constraints of~\cite{verschoor2019efficient} to simulate two cow heads colliding and a chain of rings falling. The second two simulations (Figure~\ref{fig:sim-dataset} bottom row) use the method of~\cite{Li2020IPC} to simulate a coarse mat twisting and the high speed impact of a golf ball hitting a planar wall. 

\subsection{Comparison}\label{sec:comparison}
\input{figures/ccd-table}

We compare seven state-of-the-art methods: (1) the interval root-finder (IRF)~\cite{snyder1992interval}, (2) the univariate interval root-finder (UIRF) (a special case of the rigid-body CCD from~\cite{Redon2002fast}), (3) the floating-point time-of-impact root finder~\cite{provot1997collision} (FPRF) implemented in~\cite{vouga2010AVCM}, (4) TightCCD (TCCD)~\cite{wang2015tightccd}, (5) Root Parity (RP)~\cite{brochu2012efficient}, (6) a rational implementation of Root Parity (RRP) with the degenerate cases properly handled, and (7) Bernstein Sign Classification (BSC)~\cite{tang2014fast}. For each method we collect the average query time, the number of false positives (i.e., there is no collision but the method detects one), and the number of false negatives (i.e., there is a collision but the method misses it). To obtain the ground truth we solve the multivariate CCD formulation (equations~\eqref{eq:F-vf} and \eqref{eq:F-ee}) symbolically using Mathematica~\cite{Mathematica} which takes multiple seconds per query. Table~\ref{tab:results} summarizes the results. 
Note that ``Ours'' corresponds to our new method that will be introduced and discussed in Section~\ref{sec:method} \review{and MSRF is a minimum separation CCD discussed in Section~\ref{sec:msccd-results}.}

\paragraph{IRF}
The inclusion-based root-finding described in~\cite{snyder1992interval} can be applied to both the multivariate and univariate CCD. For the multivariate case we can simply initialize the parameters of $F$ (i.e., $t, u, v$) with the size of the domain $\Omega$, evaluate $F$ and check if the origin is contained in the output interval~\cite{Snyder1993Interval}. 
If it is, we sequentially subdivide the parameters (thus shrinking the size of the intervals of $F$) until a user-tolerance $\delta$ is reached. In our comparison we use $\delta=10^{-6}$.
The major advantage of this approach is that it is guaranteed to be conservative: it is impossible to shrink the interval of $F$ to zero. A second advantage is that a user can easily trade accuracy (number of false positives) for efficiency by simply increasing the tolerance $\delta$ \review{(Appendix~\ref{app:ib-params})}. The main drawback is that bisecting $\Omega$ in the three dimensions makes the algorithm slow, and the use of interval arithmetic further increases the computational cost and prevents the use of certain compiler optimization techniques (such as instruction reordering). We implement this approach using the numerical type provided by the Boost interval library~\cite{boost}.

\paragraph{UIRF}
\cite{snyder1992interval} can also be applied to the univariate function in Equation~\eqref{eq:univariate} by using the same subdivision technique on the single variable $t$ (as in~\cite{Redon2002fast} but for linear trajectories). The result of this step is an interval containing the earliest root in $t$ which is then plugged inside a geometric predicate to check if the primitives intersect in that interval. While finding the roots with this approach might, at a first glance, seem easier than in the multi-variate case and thus more efficient, this is not the case in our experiments. If the polynomial has infinite roots, this algorithm will have to refine the entire domain to the maximal allowed resolution, and check the validity of each interval, making it correct but very slow on degenerate cases \review{(Appendix~\ref{app:ib-params})}. This results in a longer average runtime than its multivariate counterpart. Additionally, it is impossible to control the accuracy of the other two parameters (i.e., $u, v$), thus introducing more false positives.

\paragraph{FPRF}
\citet{vouga2010AVCM} aim to solve the univariate CCD problem using only floating-point computation. To mitigate false negatives, the method uses a numerical tolerance \review{$\eta$ (Appendix~\ref{app:fprf-msccd}) shows how $\eta$ affects running time, the false positive, and negative).} The major limitations are that the number of false positives cannot be directly controlled as it depends on the relative position of the input primitives and that false negatives can appear if the parameter is not tuned accordingly to the objects velocity and scale. Additionally, the reference implementation does not handle the edge-edge CCD when the two edges are parallel. This method is one of the fastest, which makes it a very popular choice in many simulation codes.


\paragraph{TCCD}
TightCCD is a conservative floating-based implementation of~\citet{tang2014fast}. It uses the univariate formulation coupled with three inequality constraints (two for the edge-edge case) to ensure that the univariate root is a CCD root. The algorithm expresses the cubic polynomial $f$ as a product and sum of three low order polynomials in Bernstein form. With this reformulation the CCD problem becomes checking if univariate Bernstein polynomials are positive, which can be done by checking some specific points. This algorithm is extremely fast but introduces many false positives which are impossible to control. In our benchmark, this is the only non-interval method without false negatives. The major limitation of this algorithm is that it \emph{always} detects collision if the primitives are moving in the same plane, independently from their relative position.

\paragraph{RP and RRP}
These two methods use the multivariate formulation $F$ (equations~\eqref{eq:F-vf} and \eqref{eq:F-ee}). The main idea is that the \emph{parity} of the roots of $F$ can be reduced to a ray casting problem. Let $\partial\Omega$ be the boundary of $\Omega$, the algorithm shoots a ray from the origin and counts the parity of the intersection between the ray and $F(\partial\Omega)$ which corresponds to the parity of the roots of $F$. Parity is however insufficient for CCD: these algorithms cannot differentiate between zero roots (no collision) and two roots (collision), since they have the same parity. We note that this is a rare case happening only with sufficiently large time-steps and/or velocities: we found  13 (handcrafted dataset) and 7 (simulation dataset) queries where these methods report a false negative. 

We note that the algorithm described in~\cite{brochu2012efficient} (and its reference implementation) does not handle some degenerate cases leading to both false negatives and positives. For instance, in Appendix~\ref{app:britson-bug}, we show an example of a ``hourglass'' configuration where RP misses the collision, generating a false negative. To overcome this limitations and provide a fair comparison to these techniques, we implemented a na\"ive  version of this algorithm that handles all the degenerate cases using rational numbers to simplify the coding (see the additional materials). We opted for this rational implementation since properly handling the degeneracies using floating-point requires designing custom higher precision predicates for all cases. The main advantage of this method is that it is exact (when the degenerate cases are handled) as it does not contain any tolerance and thus has zero false positives. We note that the runtime of our rational implementation is extremely high and not representative of the runtime of a proper floating point implementation of this algorithm.

\paragraph{BSC}
This efficient and exact method uses the univariate formulation coupled with inequality constraints to ensure that the coplanar primitives intersects. The coplanarity problem reduces to checking if $f$ in Bernstein form has a root. \citet{tang2014fast} explain how this can be done \emph{exactly} by classifying the signs of the four coefficients of the cubic Bernstein polynomial. The classification holds only if the cubic polynomial has monotone curvature; which can be achieved by splitting the curve at the inflection point. This splitting, however, cannot be computed exactly as it requires divisions (Appendix~\ref{app:dinesh-bug}). In our comparison, we modified the reference implementation to fix a minor typo in the code and to handle $f$ with inflection points by \emph{conservatively} reporting collision. This change introduces potential false positives, and we refer to the additional material for more details and for the patch we applied to the code.


\paragraph{Discussion and Conclusions}

From our extensive benchmark of CCD algorithms, we observe that most algorithms using the univariate formulation have false negatives. While the reduction to univariate root findings provides a performance boost, filtering the roots (without introducing false positives) is a challenging problem for which a robust solution is still elusive.

Surprisingly, only the oldest method, IRF, is at the same time reasonably efficient (e.g., it does not take multiple seconds per query as Mathematica), correct (i.e., no false negatives), and returns a small number of false positives (which can be controlled by changing the tolerance $\delta$). It is however slower than other state of the art methods, which is likely the reason why it is currently not widely used. In the next section we show that it is possible to change the inclusion function used by this algorithm to keep its favorable properties, while decreasing its runtime by ${\sim}250$ times, making its performance competitive with state of the art methods.

%% file: figures/ccd-table.tex
\begin{table}
    \caption{Summary of the average runtime in ${\mu}s$ (t), number of false positive (FP), and number of false negative (FN) for the six competing methods. }
    \centering\scriptsize
     \review{\Handcrafted Dataset (12K)} -- Vertex-Face CCD
    \begin{tabular}{l|cccccccc|c}
\toprule
{} &      IRF &      UIRF &  FPRF &  TCCD &    RP &     RRP &     BSC &  MSRF &   Ours \\
\midrule
t  &  14942.40 &  124242.00 &  2.18 &  0.38 &  1.41 &  928.08 &  176.17 &  12.90 &  1532.54 \\
FP &        87 &        146 &     9 &   903 &     3 &       0 &      11 &     16 &      108 \\
FN &         0 &          0 &    70 &     0 &     5 &       5 &      13 &    386 &        0 \\
\bottomrule
\end{tabular}
    \\[1.5em]
    
    \review{\Handcrafted Dataset  (15K)}-- Edge-Edge CCD
    \begin{tabular}{l|cccccccc|c}
\toprule
{} &       IRF &      UIRF &  FPRF &  TCCD &    RP &      RRP &    BSC &  MSRF &     Ours \\
\midrule
t  &  12452.60 &  18755.80 &  0.48 &  0.33 &  2.33 &  1271.32 &  121.80 &  2.72 &  3029.83 \\
FP &       141 &       268 &     5 &   404 &     3 &        0 &      28 &    14 &      214 \\
FN &         0 &         0 &   147 &     0 &     8 &        8 &      47 &   335 &        0 \\
\bottomrule
\end{tabular}
    \\[1.5em]
    
    Simulation Dataset (18M) -- Vertex-Face CCD
    \begin{tabular}{l|cccccccc|c}
\toprule
{} &     IRF &     UIRF &  FPRF &   TCCD &    RP &      RRP &    BSC &   MSRF &  Ours \\
\midrule
t  &  115.89 &  6191.98 &  7.53 &   0.24 &  0.25 &  1085.13 &  34.21 &  51.07 &  0.74 \\
FP &       2 &       18 &     0 &  95638 &     0 &        0 &  23015 &     75 &     2 \\
FN &       0 &        0 &  5184 &      0 &     0 &        0 &      0 &      0 &     0 \\
\bottomrule
\end{tabular}
    \\[1.5em]

    Simulation Dataset (41M) -- Edge-Edge CCD
\begin{tabular}{l|cccccccc|c}
\toprule
{} &     IRF &    UIRF &  FPRF &   TCCD &    RP &      RRP &    BSC &   MSRF &  Ours \\
\midrule
t  &  215.80 &  846.57 &  0.23 &   0.23 &  0.37 &  1468.70 &  12.87 &  10.39 &  0.78 \\
FP &      71 &   16781 &     0 &  82277 &     0 &        0 &   4593 &    228 &    17 \\
FN &       0 &       0 &  2317 &      0 &     7 &        7 &     27 &      1 &     0 \\
\bottomrule
\end{tabular}
    \label{tab:results}
\end{table}

%% file: 04-method.tex
\section{Method}\label{sec:method}

We describe the seminal bisection root-finding algorithm introduced in~\cite{snyder1992interval} (Section~\ref{sec:snyder}) and then introduce our novel Boolean inclusion function and how to evaluate it exactly and efficiently using floating point filters (Section~\ref{sec:inclusion}).

\subsection[Solve Algorithm Snyder 1992]{\textsc{Solve} Algorithm~\cite{snyder1992interval}}\label{sec:snyder}

An interval $i = [a, b]$ is defined as 
\[
i = [a,b] = \{x|a \leq x\leq b, x, a, b \in \RR\},
\]
and, similarly, an $n$-dimensional interval is defined as
\[
I= i_1 \times\dots\times i_n,
\]
where $i_k$ are intervals. We use $\LE{i}$ and $\RI{i}$ to refer to the left and right parts of an unidimensional interval $i$.
The width of an interval, written as $w(i) = w([\LE{i},\RI{i}])$, is defined by 
\[
w(i)=\LE{i}-\RI{i}
\]
and similarly, the width of an $n$-dimensional interval
\[
w(I)=\max_{k=\{1,\dots,n\}}w(i_k).
\]

An interval can be used to define an inclusion function. Formally, given an $m$-dimensional interval $D$ and a continuous function $g\colon \RR^m\to\RR^n$, 
an inclusion function for $g$, written $\incl g$, is a function such that
\[
\forall x \in D \quad g(x) \in \incl g(D).
\]
In other words, $\incl g(D)$ is a $n$-dimensional interval bounding the range of $g$ evaluated over an $m$-dimensional interval $D$ bounding its domain. We call the inclusion function $\incl g$ of a continuous function $g$ convergent if for an interval $X$
\[
 w(X) \to 0 \implies w\big(\incl g(X)\big) \to 0.
\]


\begin{figure}
    \centering
    \includegraphics[width=.33\linewidth]{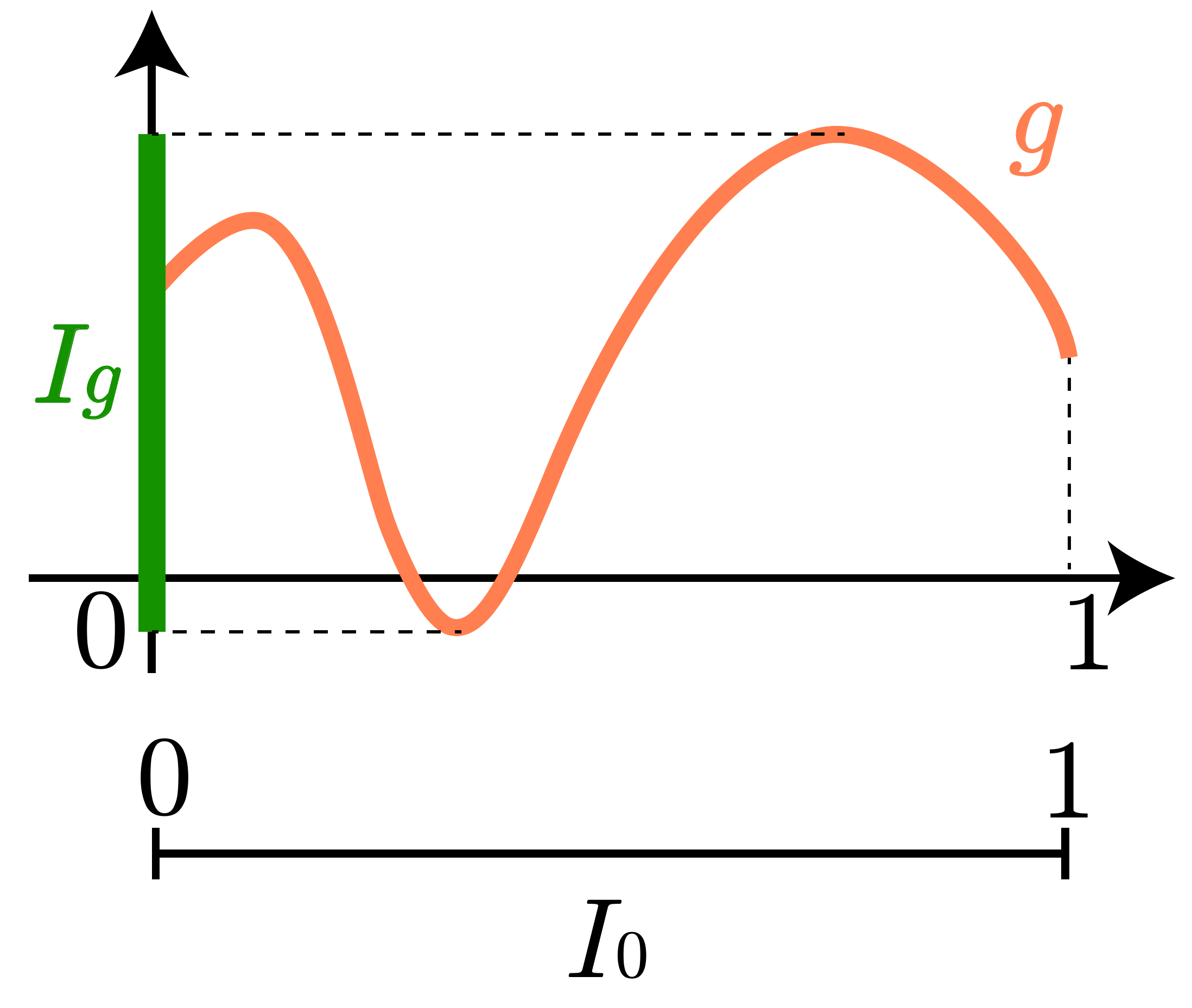}\hfill
    \includegraphics[width=.33\linewidth]{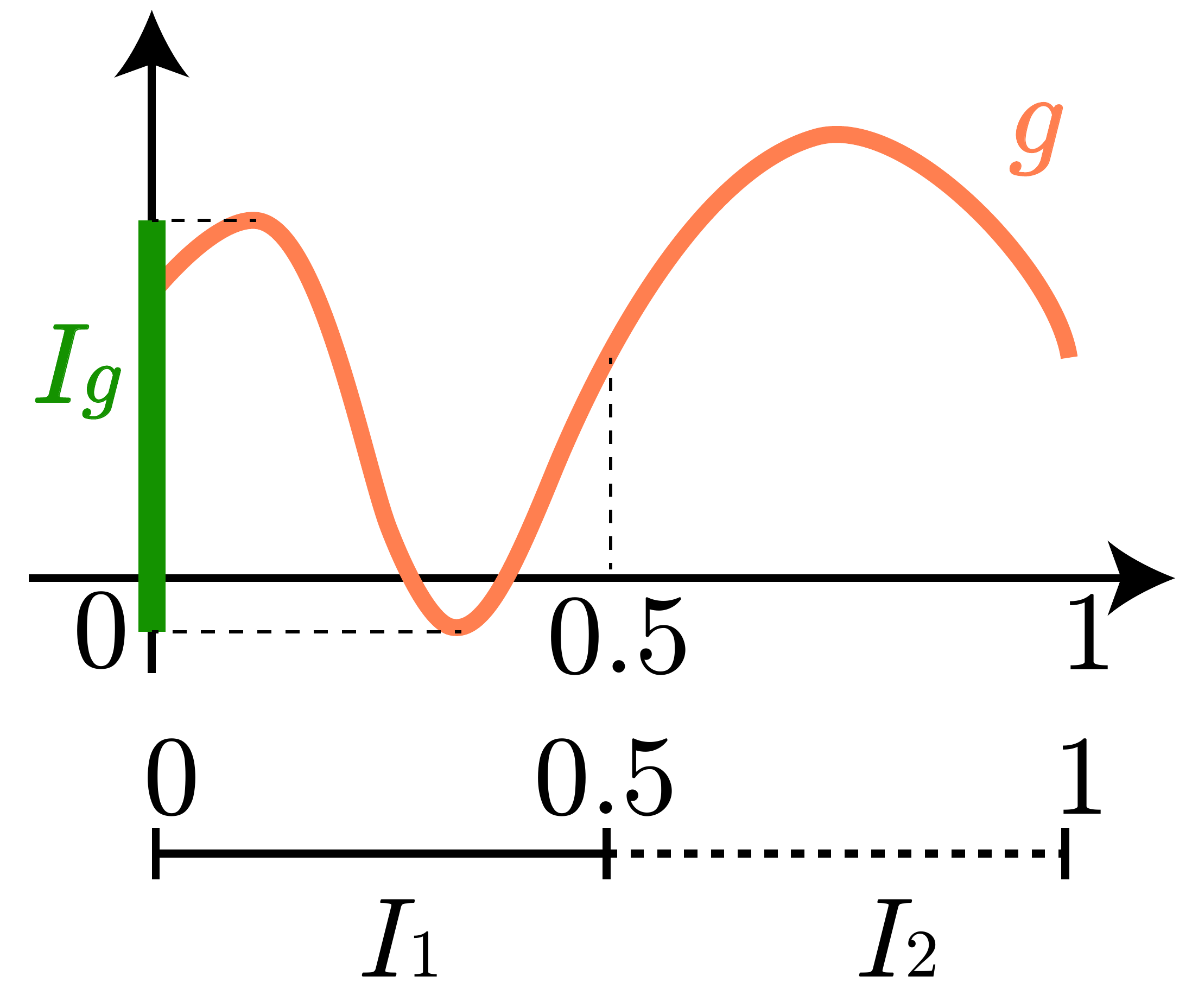}\hfill
    \includegraphics[width=.33\linewidth]{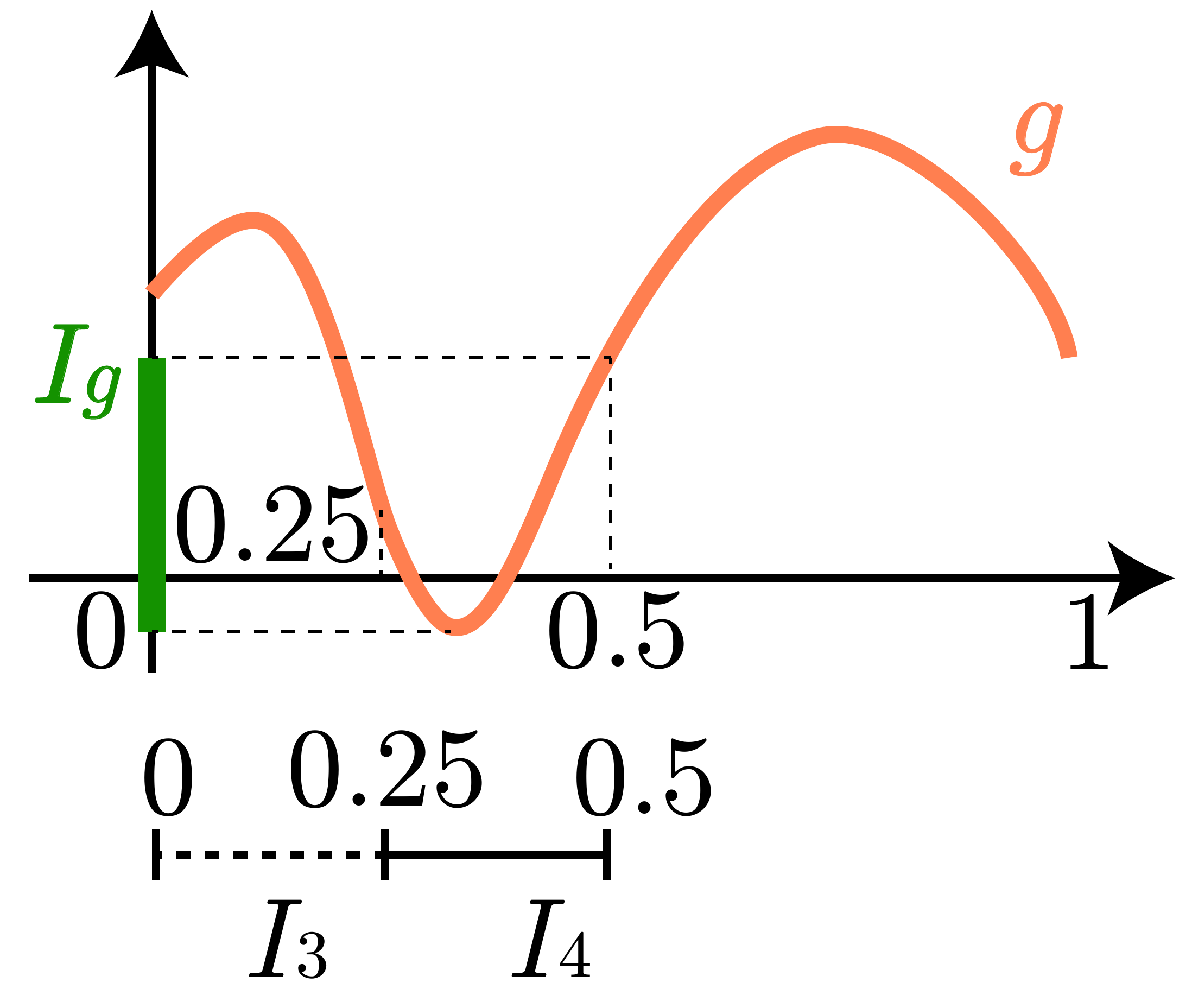}\par
    \parbox{.33\linewidth}{\centering\scriptsize{$\ell=0$}}\hfill
    \parbox{.33\linewidth}{\centering\scriptsize{$\ell=1$}}\hfill
    \parbox{.33\linewidth}{\centering\scriptsize{$\ell=2$}}\par
    \caption{\review{1D illustration of the first three levels of the inclusion based root-finder in~\cite{snyder1992interval}.}}
    \label{fig:bisection}
\end{figure}

A \emph{convergent} inclusion function can be used to find a root of a function $g$ over a domain bounded by the interval $I_0=[\LE{x_1},\RI{x_1}]\times\dots\times[\LE{x_m},\RI{x_m}]$.
To find the roots of $g$, we
sequentially bisect the initial $m$-dimensional interval $I_0$, until it becomes sufficiently small (Algorithm~\ref{alg:snyder}). \review{Figure~\ref{fig:bisection} shows a 1D example (i.e., $g\colon\RR\to\RR$) of a bisection algorithm.}
The algorithm starts by initializing a \review{stack} \revieww{$S$} of intervals to be checked with $I_0$ (line~\ref{ll:init}). At every level $\ell$ (line~\ref{ll:lopp}), the algorithm retrieves an interval $I$ from \revieww{$S$} and evaluates the inclusion function to obtain the interval $I_g$ (line~\ref{ll:inc}). Then it checks if the root is included in $I_g$ (line~\ref{alg:line:origin-in}). If not $I$ can be safely discarded since $I_g$ bounds the range of $g$ over the domain bounded by $I$. Otherwise ($0 \in I_g$), it checks if $w(I)$ is smaller than a user-defined threshold $\delta$. If so it appends $I$ to the \revieww{result} (line~\ref{ll:add-sol}). If $I$ is too large, the algorithm splits one of its dimensions (e.g., $[\LE{x_1},\RI{x_1}]$ is split in $[\LE{x_1},\tilde x_1]$ and $[\tilde x_1, \RI{x_1}]$ with $\tilde x_1 = (\LE{x_1}+\RI{x_1})/2$) and appends the two new intervals $I_1, I_2$ to the \revieww{stack $S$} (line~\ref{ll:append}).



\begin{algorithm}
\begin{algorithmic}[1]
\Function{solve}{$I_0, g, \delta$}
    \State $\revieww{\text{res}}\gets \emptyset$
    \State $\revieww{S}\gets \{I_0\}$\label{ll:init}
    \State $\ell \gets 0$
    \While{$L\neq\emptyset$}\label{ll:lopp}
        \State $I \gets$ \Call{pop}{$L$}
        \State $I_g \gets \incl g(I)$\label{ll:inc}\Comment{Compute the inclusion function}
        \If{$0 \in I_g$}\label{alg:line:origin-in}
            \If{$w(I) < \delta$}\Comment{$I$ is small enough}
                \State $\revieww{\text{res}}\gets R\cup \{I\}$\label{ll:add-sol}
            \Else
                \State $I_1, I_2 \gets \Call{split}{I}$
                \State $\revieww{S}\gets \revieww{S} \cup \{I_1, I_2\}$\label{ll:append}
            \EndIf
        \EndIf
        \State $\ell \gets \ell +1$
    \EndWhile
    
    \Return \revieww{\text{res}}
\EndFunction
\end{algorithmic}
    
    \caption{Inclusion-based root-finder}
    \label{alg:snyder}
\end{algorithm}

\paragraph{Generic Construction of Inclusion Functions}

\citet{snyder1992interval} proposes the use of interval arithmetic as a universal and automatic way to build inclusion functions for arbitrary expressions. However, interval arithmetic adds a performance overhead to the computation. For example, the product between two intervals is
\[
[a, b] \cdot[c, d] = [\min(ac, ad, bc, bd), \max(ac, ad, bc, bd)],
\]
which requires four multiplications and two min/max instead of one multiplication. In addition, the compiler cannot optimize composite expressions, since the rounding modes need to be correctly set up and the operation needs to be executed in order to avoid rounding errors~\cite{boost}. 

\subsection{Predicate-Based Bisection Root Finding}\label{sec:inclusion}

Instead of using interval arithmetic to construct the inclusion function $\incl F$ for the interval 
$
I_\Omega = I_t \times I_u \times I_v = [0, 1] \times [0, 1] \times [0, 1]
$
around the domain $\Omega$, we propose to define an inclusion function tailored for $F$ (both for Equation \eqref{eq:F-vf} and \eqref{eq:F-ee}) as the box
\begin{equation}\label{eq:box}
B_F(I_\Omega) = [m^x, M^x] \times [m^y, M^y] \times [m^z, M^z]
\end{equation}
with
\[
m^c = \min_{i=1,\dots,8}(\v_i^c),\quad
M^c = \max_{i=1,\dots,8}(\v_i^c),\quad
\review{c=\{x,y,z\}}
\]
\[
\v_i = F(t_m, u_n, v_l),\quad t_m, u_n, v_l \in \{0, 1\},\text{ and}\quad m,n,l\in\{1,2\}.
\]

\begin{proposition}\label{th:inclusion}
The inclusion function $B_F$ defined in~\eqref{eq:box} is the \emph{tightest axis-aligned inclusion function} of $F$.
\end{proposition}
\begin{proof}
We note that for any given $\tilde{u}$ the function $F(t,\tilde{u},v)$ is bilinear; we call this function function $F_{\tilde{u}}(t,v)$. 
Thus, $F$ can be regarded as a bilinear function whose four control points move along linear trajectories $\mathcal{T}(u)_i, i=1,2,3,4$.
The range of $F_{\tilde{u}}$ is a bilinear surface which is bounded by the tetrahedron constructed by the four vertices forming the bilinear surface, which are moving on $\mathcal{T}_i$. Thus, $F$ is bounded by  every tetrahedron formed by $\mathcal{T}(u)_i$, implying that $F$ is bounded by the convex hull of the trajectories' vertices, which are the vertices $v_i, i = 1,\cdots,8$ defining $F$. Finally, since $B_F$ is the axis-aligned bounding box of the convex-hull of $v_i, i = 1,\cdots,8$, $B_F$ is an inclusion function for $F$.

Since the vertices of the convex hull belong to $F$ and the convex hull is the tightest convex hull, the bounding box $B_F$ of the convex hull is the tightest inclusion function.
\end{proof}

\begin{theorem}
The inclusion function $B_F$ defined in~\eqref{eq:box} is convergent.
\end{theorem}
\begin{proof}
We first note that $F$ is trivially continuous, second that the standard interval-based inclusion function $\incl F$ constructed with intervals is axis-aligned. Therefore, from Proposition~\ref{th:inclusion}, it follows that $B_F(I)\subseteq \incl F(I)$ for any interval $I$. Finally, since $\incl F$ is convergent~\cite{snyder1992interval}, then also $B_F$ is.
\end{proof}

The inclusion function $B_F$ turns out to be ideal for constructing a predicate: to use this inclusion function in the \textsc{solve} algorithm (Algorithm \ref{alg:snyder}), we only need to check if, for a given interval $I$, $B_F(I)$ contains the origin (line~\ref{alg:line:origin-in}). Such a Boolean predicate can be conservatively evaluated using floating point filtering.

\paragraph{Conservative Predicate Evaluation}
Checking if the origin is contained in an axis-aligned box is trivial and it reduces to checking if the zero is contained in the three intervals defining the sides of the box. 
In our case, this requires us to evaluate the sign of $F$ at the eight box corners.
However, the vertices of the co-domain are computed using floating point arithmetic and can thus be inaccurate. We use forward error analysis to conservatively account for these errors as follows.

Without loss of generality, we focus only on the $x$-axis.
Let $\{v_i^x\}, {i=1,\dots, 8}$ be the set of $x$-coordinates of the 8 vertices of the box represented in double precision floating-point numbers. The error bound for $F$ (on the $x$-axis) is
\begin{equation}\label{eq:ccd-error}
\begin{array}{c}
\varepsilon_\text{ee}^x = 6.217248937900877 \times 10^{-15} \gamma_x^3 \\
\varepsilon_\text{vf}^x = 6.661338147750939 \times 10^{-15} \gamma_x^3
\end{array}
\end{equation}
with 
\begin{equation*}
\gamma_x = \max(x_\text{max}, 1) \quad\text{and}\quad
x_\text{max} = \max_{i=1,\dots,8} (|v_i^x|).
\end{equation*}
That is, the sign of $F_\text{ee}^x$ computed using floating-point arithmetic is guaranteed to be correct if $|F_\text{ee}^x| > \varepsilon_\text{ee}^x$, and similarly for the vertex face case. If this condition does not hold, we conservatively assume that the zero is contained in the interval, thus leading to a possible false positive.
The two constants $\varepsilon_\text{ee}^x$ and $\varepsilon_\text{vf}^x$ are floating point filters for $F_\text{ee}^x$ and $F_\text{vf}^x$ respectively, and were derived using~\cite{attene20}.

\paragraph{Efficient Evaluation}
The $x,y,z$ predicates defined above depend only on a subset of the coordinates of the eight corners of $B_F(I)$. We can optimally vectorize the evaluation of the eight corners using \texttt{AVX2} instructions (${\sim}4\times$ improvement in performance), since it needs to be evaluated on eight points and all the computation is standard floating-point arithmetic. Note that we used \texttt{AVX2} instructions because newer versions still have spotty support on current processors. After the eight points are evaluated in parallel, applying the floating-point filter involves only a few comparisons. To further reduce computation, we check one axis at a time and immediately return if any of the intervals do not contain the origin.


\paragraph{Algorithm}
\begin{algorithm}
\begin{algorithmic}[1]
\Function{solve}{$F, \delta, m_I$}
    \State $n \gets 0$ \Comment{Number of check counter}
    \State $Q\gets \{\{[0, 1]^3, 0\}\}$ \Comment{Push first interval and level 0 in $Q$}\label{ll:init-ours}
    \State $\ell_p \gets -1$ \Comment{Previous checked level is -1}
    
    \While{$Q\neq\emptyset$}
        \State $I, \ell \gets$ \Call{pop}{$Q$}\Comment{Retrieve level and interval}\label{ll:retreive}
        \State $B \gets B_F(I)$\Comment{Compute the box inclusion function}\label{ll:compute_bf}
        \State $n \gets n + 1$\Comment{Increase check number}
        
        \If{$B \cap \C_\varepsilon \neq \emptyset$}\label{ll:box-check}
            \If{$\ell \neq \ell_p$}\Comment{$I$ is the first colliding interval of $\ell$}
                \State $I_f\gets I_t$\label{ll:first-box}\Comment{Save $t$-component of $I$}
            \EndIf
            \If{$n \geq m_I$}\Comment{Reached max number of checks}
                \State \Return $\LE{I_f}, w(I_t)$\label{ll:stop-max-iter}\Comment{Return left side of $I_f$}
            \EndIf
            \State
            \If{$w(B) < \delta$ or $B \subseteq \C_\varepsilon$}\label{ll:termination-check}
                \If{$\ell \neq \ell_p$}\label{ll:no-return}
                    \State \Return $\LE{I_f}, w(I_t)$\Comment{Root found}
                \EndIf
            \Else
                \State $I_1, I_2 \gets \Call{split}{I}$\label{ll:split}
                \State $Q\gets Q \cup \{\{I_1, \ell+1\}, \{I_2, \ell+1\}\}$
                \State $\Call{sort}{Q, \revieww{\textproc{order}}}$\label{ll:sort}
            \EndIf
            \State $\ell_p = \ell$\Comment{Update the previous colliding level}
        \EndIf
    \EndWhile
    \State\Return $\infty, 0$\Comment{$Q$ is empty and no roots were found}
\EndFunction
\State
\Function{split}{$I = I_t \times I_u \times I_v$}
\State Compute $\kappa_t, \kappa_u, \kappa_v$ according to \eqref{eq:kappas}
\State $c_t \gets w(I_t)\kappa_t$,
$~~~c_u \gets w(I_u)\kappa_u$,
$~~~c_v \gets w(I_v)\kappa_v$
\State $c \gets \max(c_t, c_u, c_v)$\label{ll:split-algo}
\If{$c_t = c$}\Comment{$c_t$ is the largest}
\State $I_1\gets[\LE{I_t}, (\LE{I_t} + \RI{I_t})/2] \times I_u \times I_v$,  
\State $I_2\gets[(\LE{I_t} + \RI{I_t})/2, \RI{I_t}] \times I_u \times I_v$
\ElsIf{$c_u = c$}\Comment{$c_u$ is the largest}
\State $I_1\gets I_t \times [\LE{I_u}, (\LE{I_u} + \RI{I_u})/2] \times I_v$,  
\State $I_2\gets I_t \times [(\LE{I_u} + \RI{I_u})/2, \RI{I_u}] \times I_v$
\Else \Comment{$c_v$ is the largest}
\State $I_1\gets I_t \times I_u \times [\LE{I_v}, (\LE{I_v} + \RI{I_v})/2]$,  
\State $I_2\gets I_t \times I_u \times [(\LE{I_v} + \RI{I_v})/2, \RI{I_v}]$
\EndIf
\State\Return $I_1, I_2$
\EndFunction
\State
\Function{\revieww{order}}{$\{I_1, \ell_1\}, \{I_2, \ell_2\}$}
\If{$\ell_1 = \ell_2$}
    \State\Return $I_1^t < I_2^t$
\Else
    \State\Return $\ell_1 < \ell_2$
\EndIf
\EndFunction
\end{algorithmic}
\caption{Complete overview of our CCD algorithm.}
\label{alg:our}
\end{algorithm}

We describe our complete algorithm in pseudocode in Algorithm~\ref{alg:our}. The input to our algorithm are the eight points representing two primitives (either vertex-face or edge-edge), a user-controlled numerical tolerance $\delta > 0$ (if not specified otherwise, in the experiment we use the default value $\delta=10^{-6}$), and the maximum number of checks $m_I >0$ (we use the default value $m_I=10^6$). \review{These choice are based on our empirical results (figures~\ref{fig:tolerance} and~\ref{fig:max-checks}).}
The output is a conservative estimate of the earliest time of impact or infinity if the two primitives do not collide in the time intervals coupled with the reached tolerance. 

Our algorithm iteratively checks the box $B=B_F(I)$, with $I=I_t\times I_u \times I_v$ = $[t_1, t_2]\times[u_1, u_2]\times[v_1, v_2]\subset I_\Omega$ (initialized with $[0,1]^3$). To guarantee a uniform box size while allowing early termination of the algorithm, we explore the space in a breadth-first manner and record the current explored level $\ell$ (line~\ref{ll:retreive}). Since our algorithm is designed to find the earliest time of impact, we sort the visiting queue $Q$ with respect to time (line~\ref{ll:sort}).

At every iteration we check if $B$ intersects the cube $\C_\varepsilon = [-\varepsilon^x, \varepsilon^x]\times [-\varepsilon^y, \varepsilon^y]\times[-\varepsilon^z, \varepsilon^z]$ (line~\ref{ll:box-check}); if it does not, we can safely ignore $I$ since there are no collisions.

\begin{figure}
    \centering
    \parbox{.55\linewidth}{\centering\includegraphics[width=0.9\linewidth]{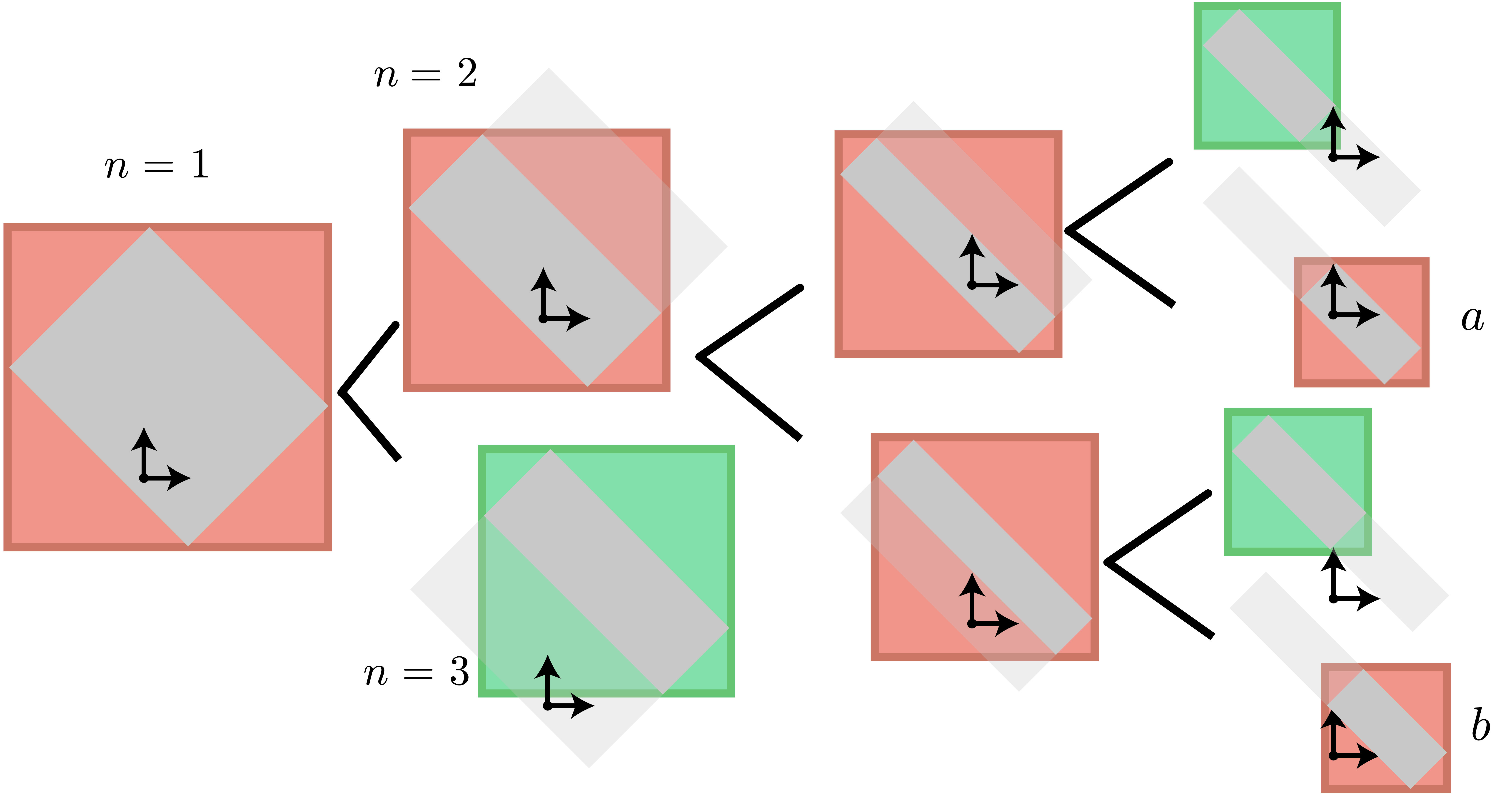}}\hfill
    \parbox{.44\linewidth}{\centering\includegraphics[width=0.9\linewidth]{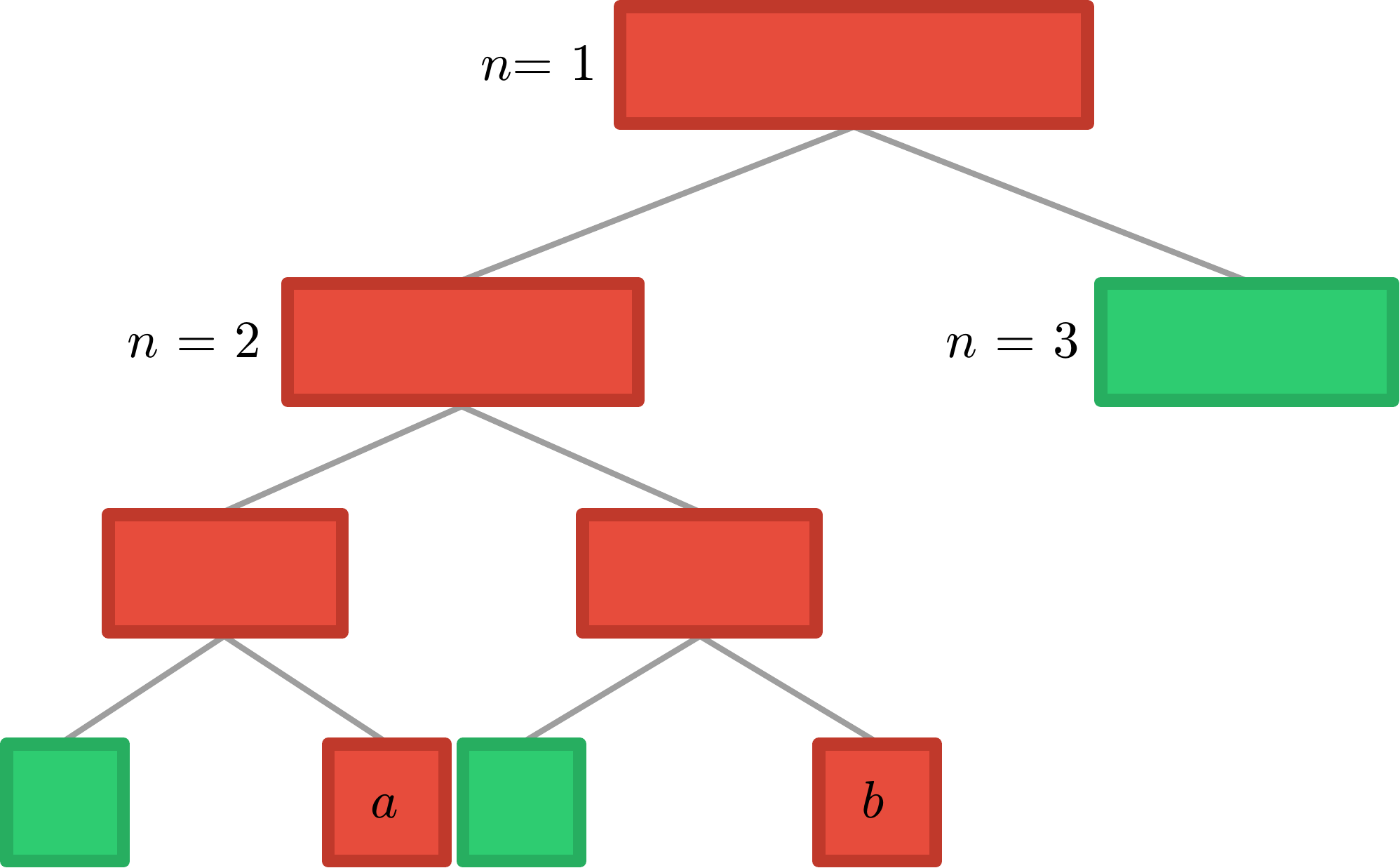}}\par
    
    \caption{\review{\revieww{A 2D example of root finding (left) and its corresponding diagram (right).} A small colliding (red) box $b$ that is not the earliest, since another box $a$ exists in the same level ($a$ did not trigger the termination of the algorithm since it is too big).}}
    \label{fig:toi-problem-left}
\end{figure}

\begin{figure}
    \parbox{.55\linewidth}{\centering\includegraphics[width=0.9\linewidth]{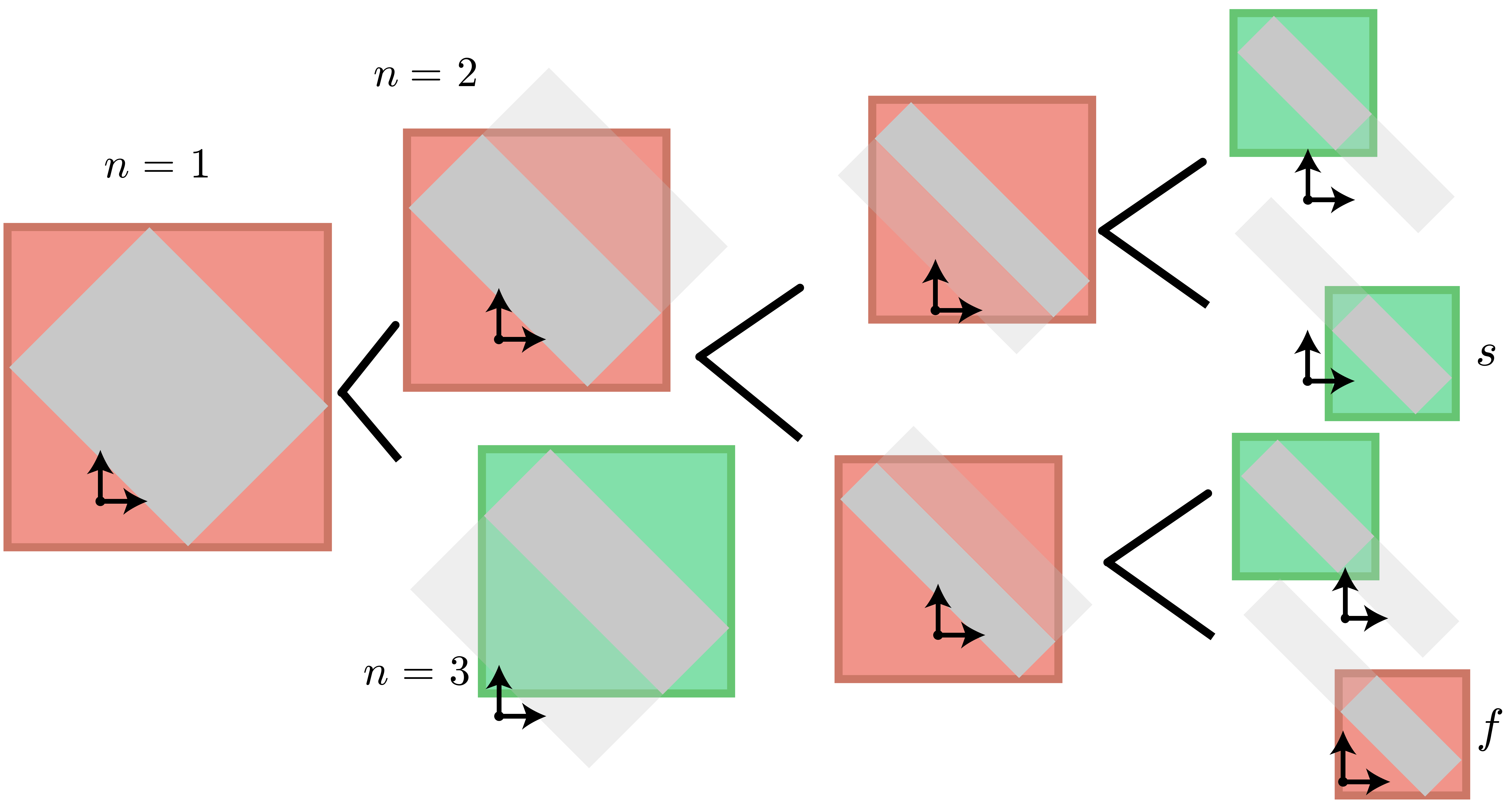}}\hfill
    \parbox{.44\linewidth}{\centering\includegraphics[width=0.9\linewidth]{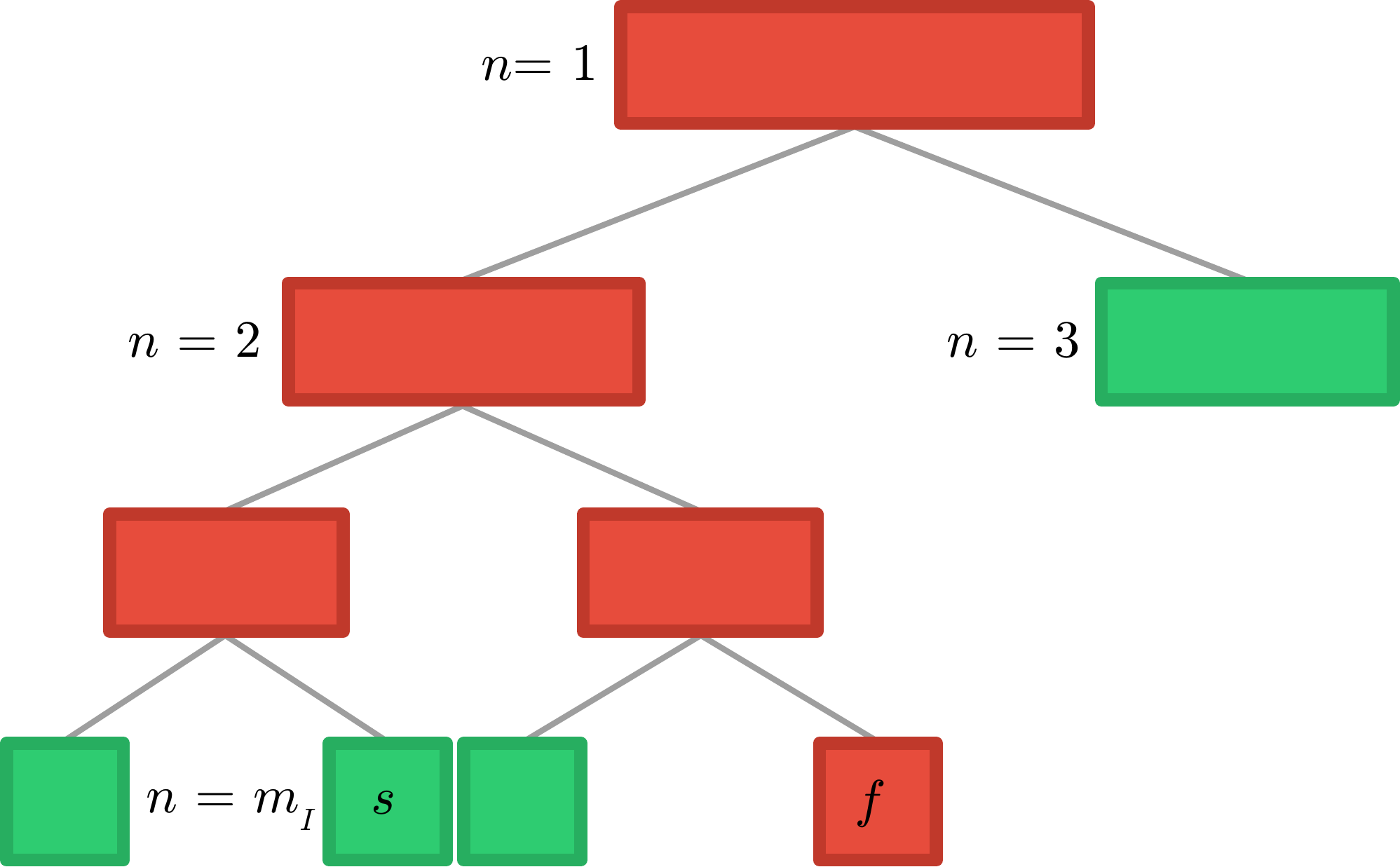}}\par
    
    \caption{\review{\revieww{A 2D example of root finding (left) and its corresponding diagram (right).}
    Our algorithm stops when the number of checks $n$ reaches $m_I$ after checking the box $s$, which is a non-colliding box (green). The algorithm will return the first colliding box ($f$) of the same level, right.}}
    \label{fig:toi-problem-right}
\end{figure}

If $B \cap \C_\varepsilon \neq \emptyset$, we first check if $w(B) < \delta$ or if $B$ is contained inside the $\varepsilon$-box (line~\ref{ll:termination-check}). In this case, it is unnecessary to refine the interval $I$ more since it is either already small enough (if $w(B) < \delta$) or any refinement will lead to collisions (if $B \subseteq \C_\varepsilon$). We return $I_t^l$ (i.e., the left hand-side of the $t$ interval of $I$) only if $I$ was the first intersecting interval of this current level (line~\ref{ll:no-return}). If $I$ is not the first intersecting in the current level, there is an intersecting box (which is larger than $\delta$) with an earlier time since the queue is sorted according to time (Figure~\ref{fig:toi-problem-left}). 

If $B$ is too big we split the interval $I$ in two sub-intervals and push them to \review{the priority queue} $Q$ (line~\ref{ll:split}). \revieww{Note that, differently from Algorithm~\ref{alg:snyder}, we use a \emph{priority queue} $Q$ instead of the stack $S$.} For the vertex-triangle CCD, the domain $\Omega$ is a prism, thus, after spitting the interval (line~\ref{ll:split}), we append $I_1, I_2$ to $Q$ only if they intersect with $\Omega$.
To ensure that $B$ shrinks uniformly \review{(since the termination criteria, Line~\ref{ll:termination-check}, is $w(B) < \delta$)}  we \emph{conservatively} estimate the width of $B$ \review{(in the codomain)} from the widths of the \review{domain's (i.e., where the algorithm is acting)} intervals $I_t, I_u, I_v$:
\begin{equation}
\label{eq:codomain}
\alpha>0, w(I_t) < \frac{\alpha}{\kappa_t}, w(I_u) < \frac{\alpha}{\kappa_u}, w(I_v) < \frac{\alpha}{\kappa_v} \implies w(B_F(I)) < \alpha
\end{equation}
with $\alpha$ a given constant and 
\begin{equation}\label{eq:kappas}
\begin{aligned}
\kappa_t &=3 \max_{i,j=1,2}\| F(0,u_i,v_j)-F(1,u_i,v_j)\|_\infty,\\
\kappa_u &=3 \max_{i,j=1,2}\| F(t_i,0,v_j)-F(t_i,1,v_j)\|_\infty,\\
\kappa_v &=3 \max_{i,j=1,2}\| F(t_i,u_j,0)-F(t_i,u_j,1)\|_\infty.
\end{aligned}
\end{equation}
\begin{proposition}
Equation~\ref{eq:codomain} holds for any positive constant $\alpha$.
\end{proposition}

\begin{proof}
While $B_F(I)$ is an interval, for the purpose of the proof we equivalently define it as an axis-aligned bounding box whose eight vertices are $b_i$. We will use the super-script notation to refer to the $x,y,z$ component of a 3D point (e.g., $b_i^x$ is the $x$-component of $b_i$) and define the set $\mathcal{I}=\{1,\dots,8\}$. By using the box definition the width of $B_F(I)$ can be written as
\[
w(B_F(I)) = \|b_M-b_m\|_\infty
\]
with
\[
b_M^k=\max_{i\in\mathcal{I}}(b_i^k)\quad\text{and}\quad
b_m^k=\min_{i\in\mathcal{I}}(b_i^k).
\]
Since $B_F(I)$ is the tightest axis-aligned inclusion function (Proposition~\ref{th:inclusion}) 
\[
b_M^k \leq \max_{i\in\mathcal{I}}{v_i^k}, \quad
b_m^k \leq \min_{i\in\mathcal{I}}{v_i^k},
\]
where 
$
v_i = F(I_t^j, I_u^k, I_v^l),
$
with $j,k,l\in\{l,r\}$, thus for any coordinate $k$ we bound
\[
b_M^k-b_m^k=
\max_{i,j\in\mathcal{I}}(v_i^k-v_j^k)\leq
\max_{i,j\in\mathcal{I}}\|v_i-v_j\|_\infty.
\]

For any pair of $v_i$ and $v_j$ we have 
\[
v_i-v_j = s_1\alpha_{l,m}+s_2\beta_{n,p}+s_3\gamma_{p,q},
\]
for some indices $l,m,n,o,p,q\in\{1,2\}$ and constant $s_1,s_2,s_3 \in\{-1,0,1\}$ with
\[
\alpha_{i,j}=w(I_t)\big(F(0,u_i,v_j)-F(1,u_i,v_j)\big), 
\]
\[
\beta_{i,j}=w(I_u)\big(F(t_i,0,v_j)-F(t_i,1,v_j)\big),
\]
\[
\gamma_{i,j}=w(I_v)\big(F(t_i,u_j,0)-F(t_i,u_j,1)\big),
\]
since $F$ is linear on the edges.
We note that $\alpha_{i,j}, \beta_{i,j}$, and $\gamma_{i,j}$ are the 12 edges of the box $B_F$. We now define
\[
e_t^k = \max_{i,j\in\{1,2\}} |\alpha^k_{i,j}|,\quad
e_u^k = \max_{i,j\in\{1,2\}} |\beta^k_{i,j}|,\quad
e_v^k = \max_{i,j\in\{1,2\}} |\gamma^k_{i,j}|
\]
which allows us to bound
\[
\max_{i,j\in\mathcal{I}}\|v_i-v_j\|_\infty\leq
\|e_t+e_u+e_v\|_\infty\leq
\|e_t\|_\infty+\|e_u\|_\infty+\|e_v\|_\infty.
\]
Since 
\[
\|e_t\|_\infty \leq  w(I_t)\max_{i,j=1,2}\| F(t_1,u_i,v_j)-F(t_2,u_i,v_j)\|_\infty=w(I_t)\kappa_t/3,
\]
and similarly $\|e_u\|_\infty \leq \kappa_u/3,\|e_v\|_\infty \leq \kappa_v/3$, we have
\[
\|e_t\|_\infty+\|e_u\|_\infty+\|e_v\|_\infty\leq
\frac{w(I_t)\kappa_t+w(I_u)\kappa_u+w(I_v)\kappa_v}3
\]
Finally, from the assumption~\eqref{eq:codomain}
it follows that
\[
w(B_F(I)) 
\leq
\max_{i,j\in\mathcal{I}}\|v_i-v_j\|_\infty \leq
\|e_t\|_\infty+\|e_u\|_\infty+\|e_v\|_\infty<
\alpha.
\]

\end{proof}
Using the estimate of the width of $I_t, I_u, I_v$ we split the dimension that leads to the largest estimated dimension in the range of $F$ (line~\ref{ll:split-algo}).

\paragraph{Fixed Runtime or Fixed Accuracy.} To ensure a bounded runtime, which is very useful in many simulation applications, we stop the algorithm after an user-controlled number of checks $m_I$. To ensure that our algorithm always returns a \emph{conservative} time of impact we record the first colliding interval $I_f$ of every level (line~\ref{ll:first-box}). When the maximum number of check is reached we can safely return the latest recorded interval $I_f$ (line~\ref{ll:stop-max-iter}) (Figure~\ref{fig:toi-problem-right}). We note that our algorithm will not respect the user specified accuracy when it terminates early: if a constant accuracy is required by applications, this additional termination criteria could be disabled, obtaining an algorithm with guaranteed accuracy but sacrificing the bound on the maximal running time. \review{Note that without the termination criteria $m_I$, it is possible (while rare in our experiments) that the algorithm will take a long time to terminate, or run out of memory due to storing the potentially large list of candidate intervals $L$.}




\subsection{Results}\label{sec:our-results}
Our algorithm is implemented in C++ and uses Eigen~\cite{eigen} for the linear algebra routines (with the \texttt{-avx2} \texttt{g++} flag). We run our experiments on a 2.35 GHz AMD EPYC\texttrademark~7452. We attach the reference implementation and the data used for our experiments, which will be released publicly.

The running time of our method is comparable to the floating-point methods, while being provably correct, for any choice of parameters. For this comparison we use a default tolerance $\delta=10^{-6}$ and default number of iterations $m_I=10^6$. All queries in the simulation dataset terminate within $10^6$ checks, while for the \handcrafted dataset only {$0.25\%$} and {$0.55\%$} of the vertex-face and edge-edge queries required more than $10^6$ checks, reaching an actual maximal tolerance $\delta$ of $2.14\times10^{-5}$ and $6.41\times10^{-5}$ for vertex-face and edge-edge respectively. We note that, despite the percentages begin small, by removing $m_I$ the \handcrafted queries take {0.015774} and {0.042477} \emph{seconds} on average for vertex-face and edge-edge respectively. This is due to the large number of degenerate queries, as can be seen from the long tail in the histogram of the run-times (Figure~\ref{fig:time-historams}). We did not observe any noticeable change of running time for the simulation dataset. 

\begin{figure}
    \centering
    \parbox{0.02\linewidth}{~}\hfill
    \parbox{.48\linewidth}{\centering Vertex-Face CCD}\hfill
    \parbox{.48\linewidth}{\centering Edge-Edge CCD}\par
    \parbox{0.02\linewidth}{\centering\rotatebox{90}{\scriptsize{\# Queries}}}\hfill
    \parbox{.48\linewidth}{\centering\includegraphics[width=0.9\linewidth]{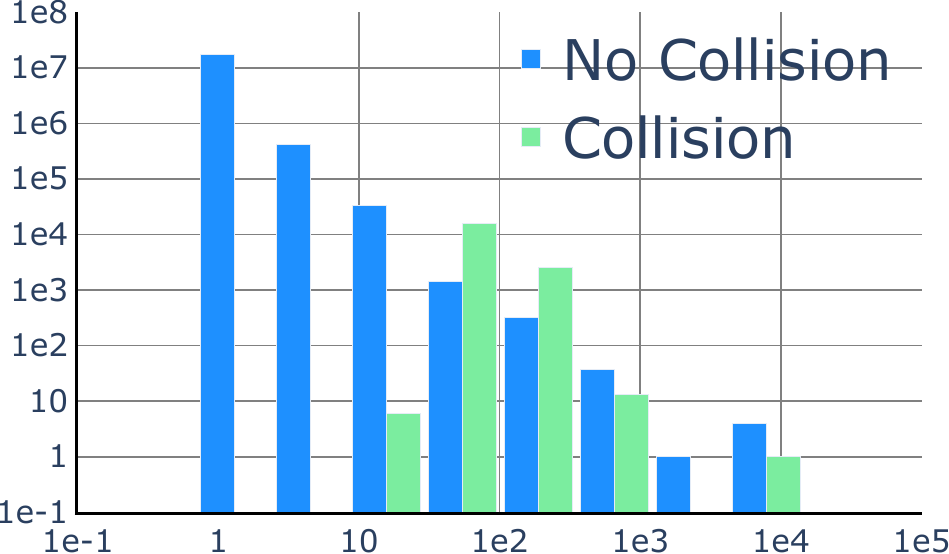}}\hfill
    \parbox{.48\linewidth}{\centering\includegraphics[width=0.9\linewidth]{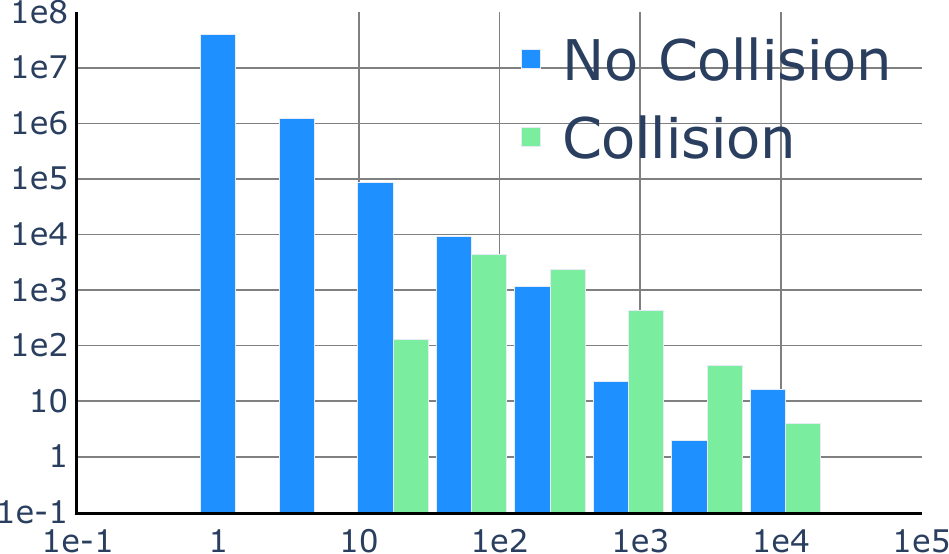}}\par
    \parbox{0.02\linewidth}{~}\hfill
    \parbox{.48\linewidth}{\centering\scriptsize{Running Time (${\mu}s$)}}\hfill
    \parbox{.48\linewidth}{\centering\scriptsize{Running Time (${\mu}s$)}}\par
    Simulation dataset\\[1em]
    \parbox{0.02\linewidth}{\centering\rotatebox{90}{\scriptsize{\# Queries}}}\hfill
    \parbox{.48\linewidth}{
    
    \centering\includegraphics[width=0.9\linewidth]{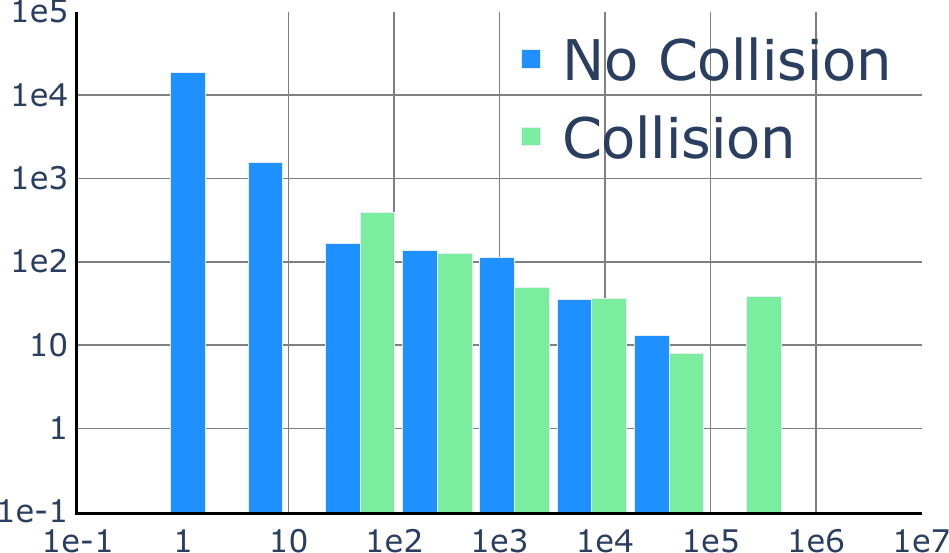}
    }\hfill
    \parbox{.48\linewidth}{
    
     \centering\includegraphics[width=0.9\linewidth]{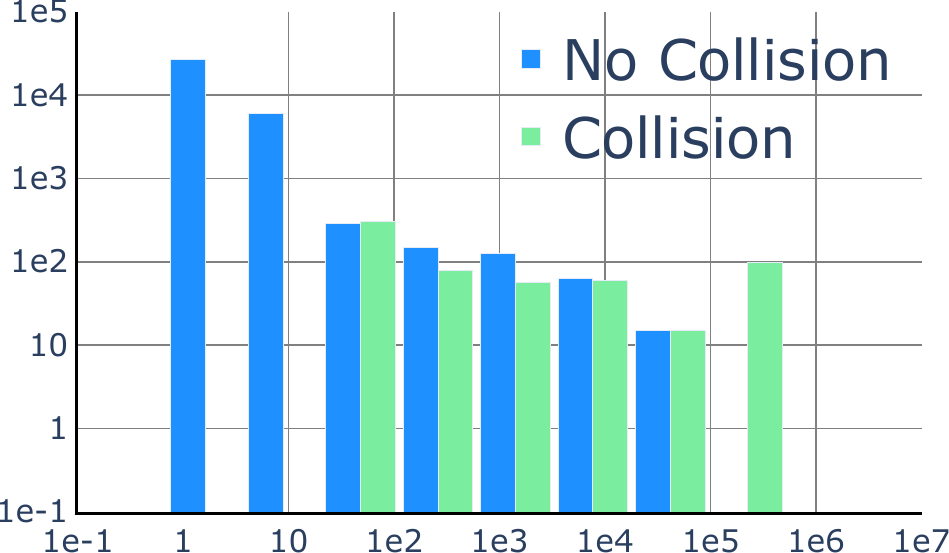}
    }\par
    \parbox{0.02\linewidth}{~}\hfill
    \parbox{.48\linewidth}{\centering\scriptsize{Running Time (${\mu}s$)}}\hfill
    \parbox{.48\linewidth}{\centering\scriptsize{Running Time (${\mu}s$)}}\par
    \Handcrafted dataset
    \caption{Log histograms of the running time of positive queries and negative queries on both dataset.}
    \label{fig:time-historams}
\end{figure}

Our algorithm has two user-controlled  parameters ($\delta$ and $m_I$) to control the accuracy and running time. The tolerance $\delta$ provides a direct control on the achieved accuracy and provides an indirect effect on the running time (Figure~\ref{fig:tolerance}). The other parameter, $m_I$, directly controls the maximal running time of each query: for small $m_I$ our algorithm will terminate earlier, resulting in a lower accuracy and thus more chances of false positives (Figure~\ref{fig:max-checks} top). We remark that, in practice, very few queries require so many subdivisions: by reducing $m_I$ to the very low value of $100$, our algorithm early-terminates only on ${\sim}0.07$\% of the 60 million queries in the simulation dataset.

\begin{figure}
    \centering
    \parbox{0.02\linewidth}{~}\hfill
    \parbox{.48\linewidth}{\centering Vertex-Face CCD}\hfill
    \parbox{.48\linewidth}{\centering Edge-Edge CCD}\par
    \parbox{0.02\linewidth}{\centering\rotatebox{90}{\scriptsize{Average Running Time ($\mu$s)}}}\hfill
    \parbox{.48\linewidth}{\centering\includegraphics[width=0.9\linewidth]{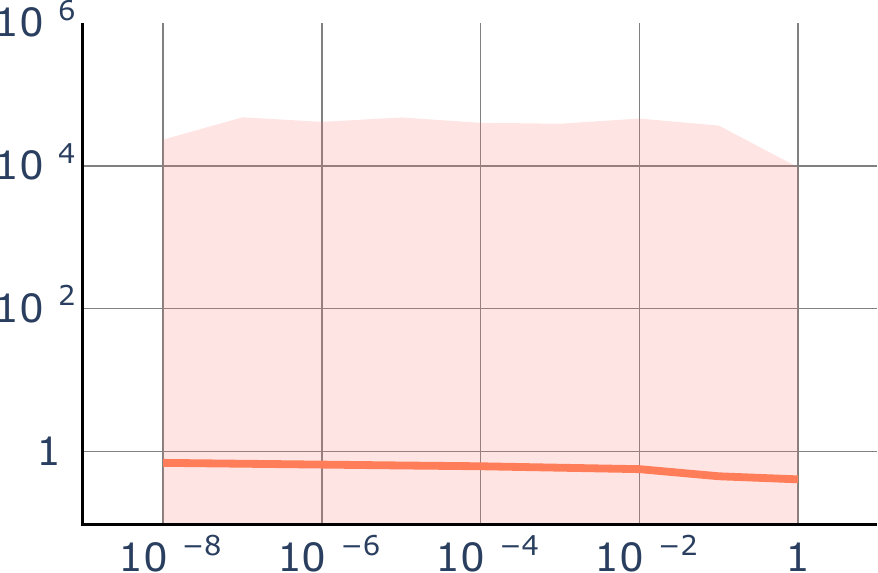}}\hfill
    \parbox{.48\linewidth}{\centering\includegraphics[width=0.9\linewidth]{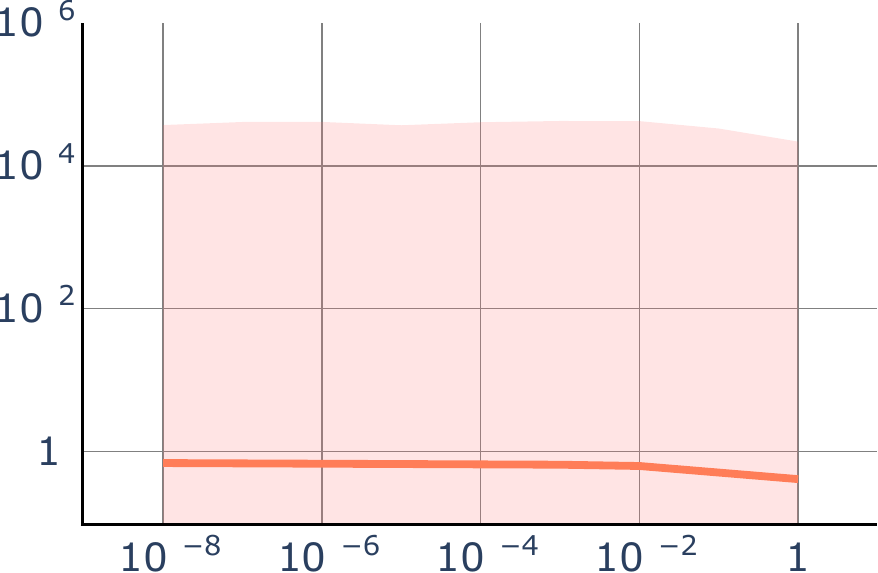}}\par
    \parbox{0.02\linewidth}{~}\hfill
    \parbox{.48\linewidth}{\centering\scriptsize{Tolerance $\delta$}}\hfill
    \parbox{.48\linewidth}{\centering\scriptsize{Tolerance $\delta$}}\\[2ex]
    \parbox{0.02\linewidth}{\centering\rotatebox{90}{\scriptsize{\# Queries}}}\hfill
    \parbox{.48\linewidth}{\centering\includegraphics[width=0.9\linewidth]{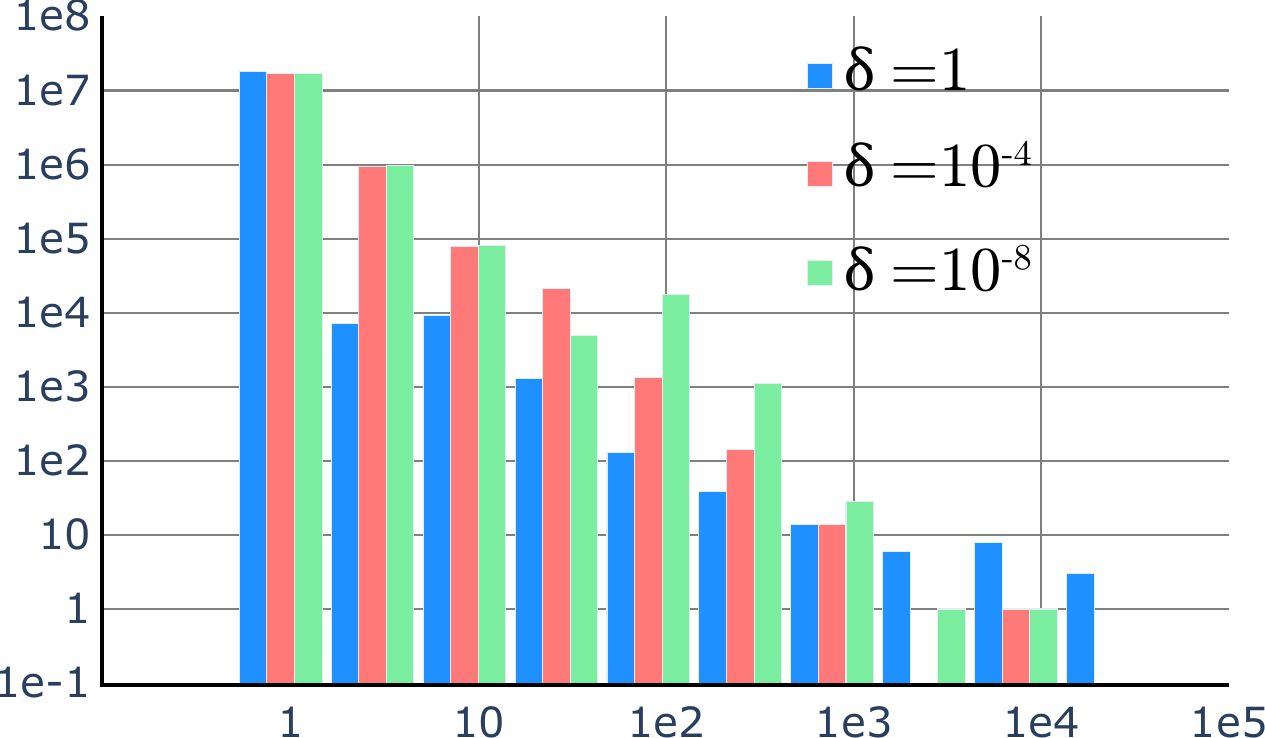}}\hfill
    \parbox{.48\linewidth}{\centering\includegraphics[width=0.9\linewidth]{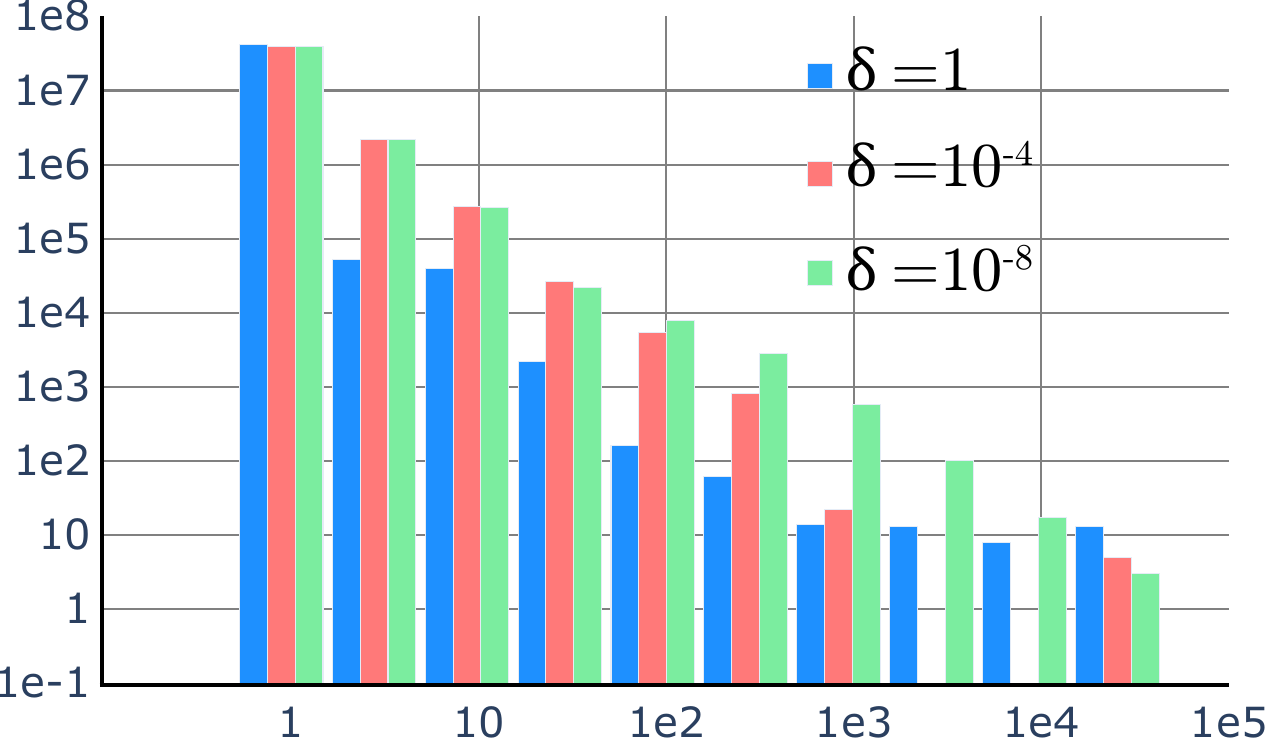}}\par
    \parbox{0.02\linewidth}{~}\hfill
    \parbox{.48\linewidth}{\centering\scriptsize{Running Time ($\mu$s)}}\hfill
    \parbox{.48\linewidth}{\centering\scriptsize{Running Time ($\mu$s)}}\par
    \caption{\review{Top, average runtime of our algorithm for different tolerances $\delta$ for the simulation dataset. The shaded area shows the range of the distribution (min and max).} \revieww{Bottom, distribution of running times of our algorithm for three different tolerances $\delta = 10^{-8}$, $10^{-4}$, and $1$ over the simulation dataset.}}
    \label{fig:tolerance}
\end{figure}

\begin{figure}
    \centering
    \parbox{0.02\linewidth}{~}\hfill
    \parbox{.48\linewidth}{\centering Vertex-Face CCD}\hfill
    \parbox{.48\linewidth}{\centering Edge-Edge CCD}\par
    \parbox{0.02\linewidth}{\centering\rotatebox{90}{\scriptsize{$\delta$ \review{max} }}}\hfill
    \parbox{.48\linewidth}{\centering\includegraphics[width=0.9\linewidth]{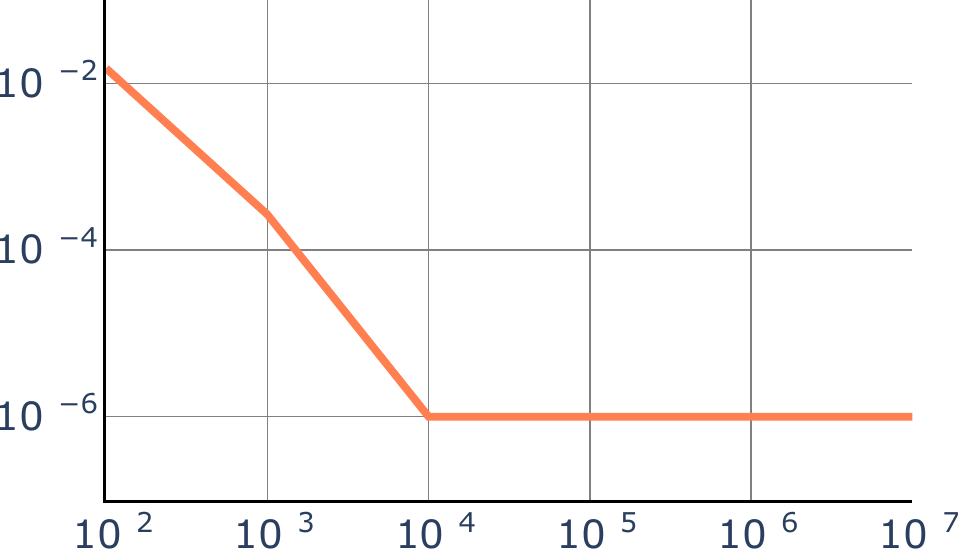}}\hfill
    \parbox{.48\linewidth}{\centering\includegraphics[width=0.9\linewidth]{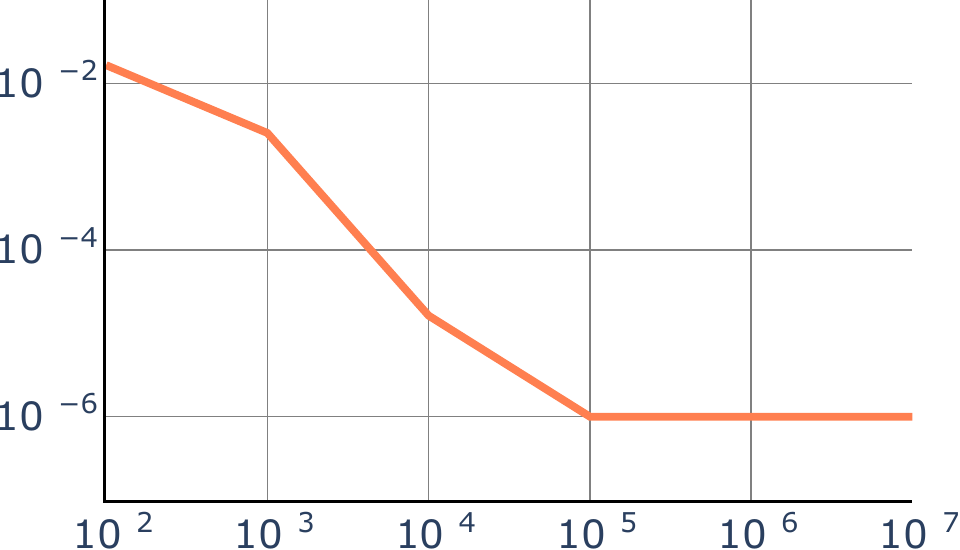}}\par
    \parbox{0.02\linewidth}{~}\hfill
    \parbox{.48\linewidth}{\centering\scriptsize{Maximum Number of Checks $m_I$}}\hfill
    \parbox{.48\linewidth}{\centering\scriptsize{Maximum Number of Checks $m_I$}}\par
    \parbox{0.02\linewidth}{\centering\rotatebox{90}{\scriptsize{Early Termination (\%)}}}\hfill
    \parbox{.48\linewidth}{\centering\includegraphics[width=0.9\linewidth]{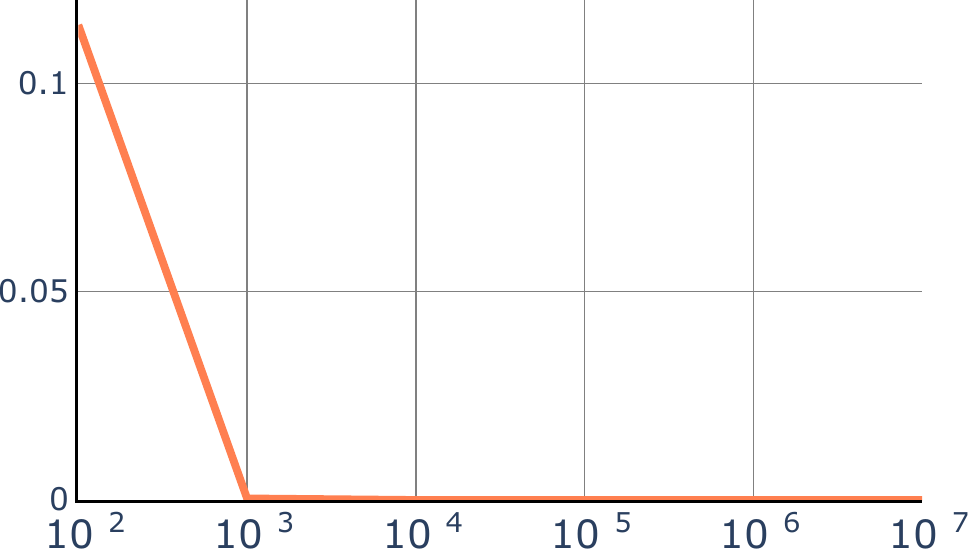}}\hfill
    \parbox{.48\linewidth}{\centering\includegraphics[width=0.9\linewidth]{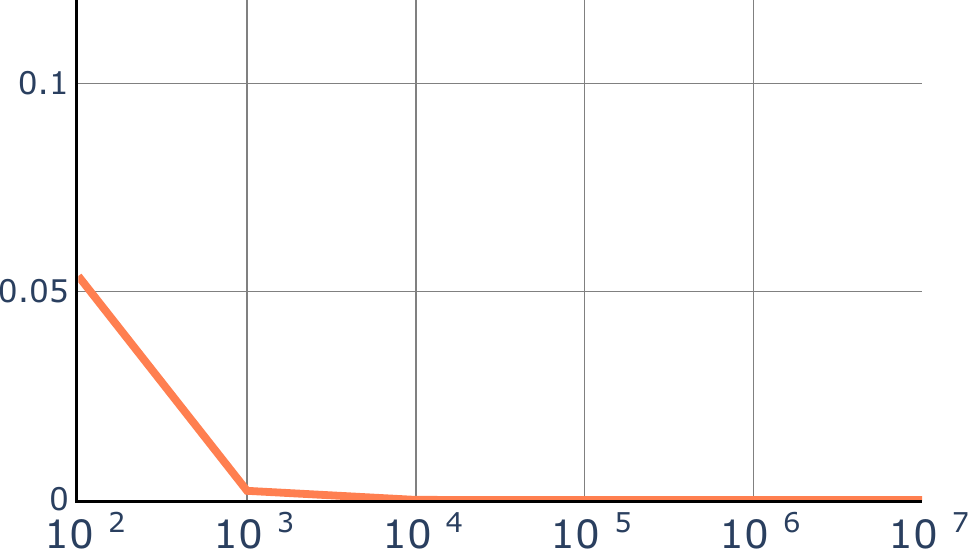}}\par
    \parbox{0.02\linewidth}{~}\hfill
    \parbox{.48\linewidth}{\centering\scriptsize{Maximum Number of Checks $m_I$}}\hfill
    \parbox{.48\linewidth}{\centering\scriptsize{Maximum Number of Checks $m_I$}}\par
    \caption{The percentage of early-termination and \review{maximum value of the} tolerance $\delta$ for different $m_I$ for the simulation dataset.}
    \label{fig:max-checks}
\end{figure}

%% file: 05.1-msccd.tex
\section{Minimum Separation CCD}\label{sec:msccd}

\input{figures/msccd-table}

An additional feature of some CCD algorithms is \emph{minimal separation}, that is, the option to report collision at a controlled distance from an object, which is used to ensure that objects are never too close. This is useful to avoid possible inter-penetrations introduced by numerical rounding after the collision response, or for modeling fabrication tolerances for additive or subtractive manufacturing. 
A minimum separation CCD (MSCCD) query is similar to a standard query: instead of checking if a point and a triangle (or two edges) are exactly overlapping, we want to ensure that they are always separated by a user-defined distance $d$ during the entire linear trajectory. Similarly to the standard CCD (Section~\ref{sec:preliminaries}) MSCCD can be express using a multivariate or a univariate formulation, usually measuring distances using the Euclidean distance. We focus on the multivariate formulation since it does not require to filter spurious roots, we refer to Section \ref{sec:comparison} for a more detailed justification of this choice.

\paragraph{Multivariate Formulation.} We observed that using the Euclidean distance leads to a challenging problem, which can be geometrically visualized as follows: the primitives will not be closer than $d$ if $F(\Omega)$ does not intersect a sphere of radius $d$ centered on the origin. This is a hard problem, since it requires checking conservatively the intersection between a sphere (which is a rational polynomial when explicitly parametrized) and $F(\Omega)$.

Studying the applications currently using minimal separation, we realized that they are not affected by a different choice of the distance function. Therefore, we propose to change the distance definition from Euclidean to Chebyshev distance (i.e., from the $L^2$ to the $L^\infty$ distance). With this minor change the problem dramatically simplifies: instead of solving for $F = 0$ (Section \ref{sec:method}), we need to solve for $|F| \leq d$. The corresponding geometric problem becomes checking if $F(\Omega)$ intersects a cube of side $2d$ centered on the origin.

\paragraph{Univariate Formulation.} The univariate formulation is more complex since it requires to redefine the notion of co-planarity for minimum separation. We remark that the function $f$ in~\eqref{eq:univariate} measures the length of the projection of $q(t)$ along the normal, thus to find point at distance $d$ the equation becomes $f(t) \leq \langle n(t), q(t)\rangle = d\|n(t)\|$. To keep the equation polynomial, remove the inequality, and avoid square roots, the univariate MSCCD root finder becomes
\[
\langle n(t), q(t)\rangle^2 - d^2 \|n(t)\|^2.
\]
We note that this polynomial becomes sextic, and not cubic as in the zero-distance version. To account for replacing the inequality with an equality, we also need to check for distance between $q$ and the edges and vertices of the triangle~\cite{harmon2011IAGM}. In addition to finding the roots of several high-order polynomials, this formulation, similarly to the standard CCD, suffers from infinite roots when the two primitives are moving on a plane at distance $d$ from each other.

\subsection{Method}\label{sec:msccd-method}

The input to our MSCCD algorithm are the same as the standard CCD (eight coordinates, $\delta$, and $m_I$) and the minimum separation distance $d \geq 0$. Our algorithm returns the earliest time of impact indicating if two primitives become closer than $d$ as measured by the $L^\infty$ norm.

\begin{figure}
    \centering
    \includegraphics[width=.33\linewidth]{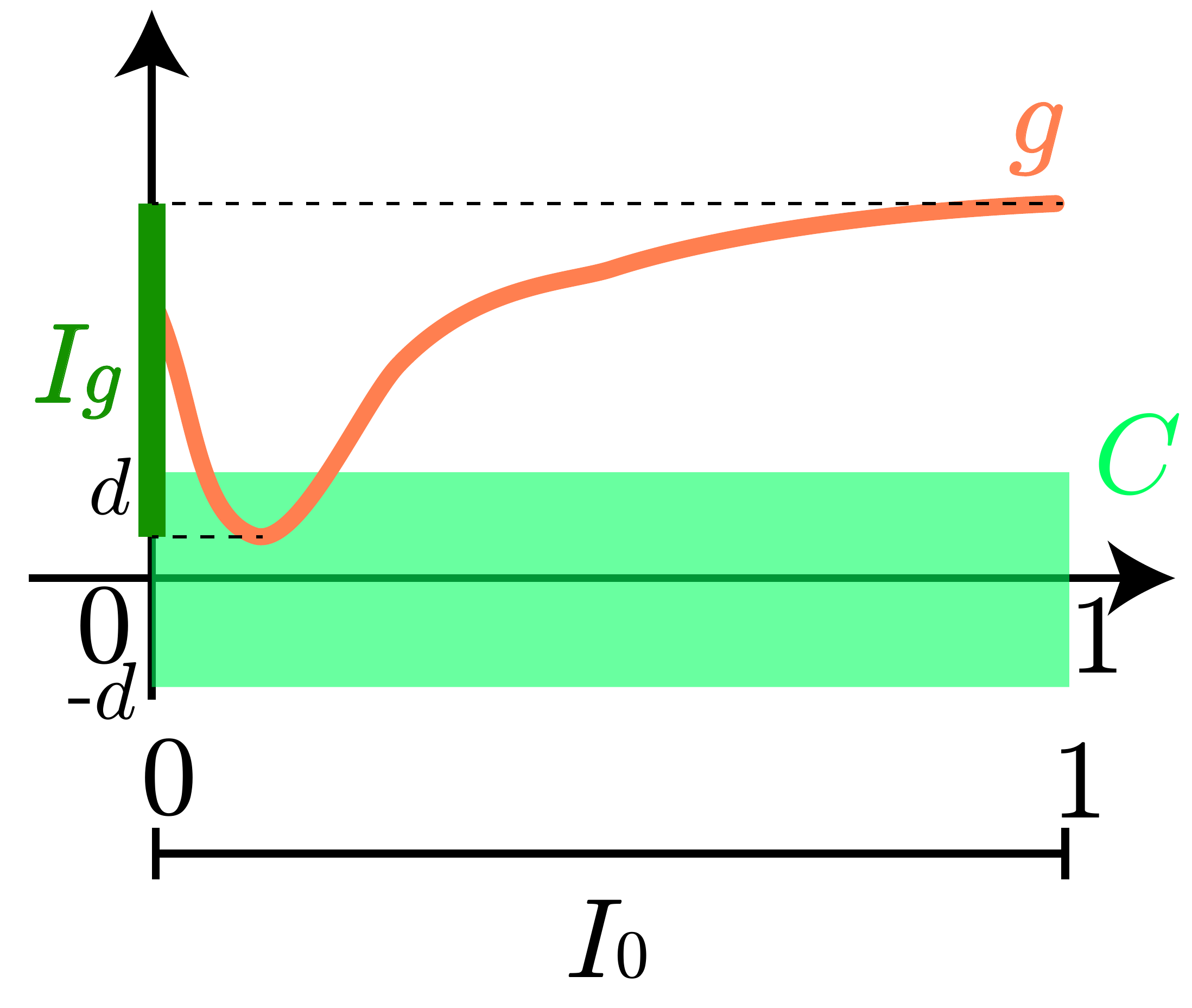}\hfill
    \includegraphics[width=.33\linewidth]{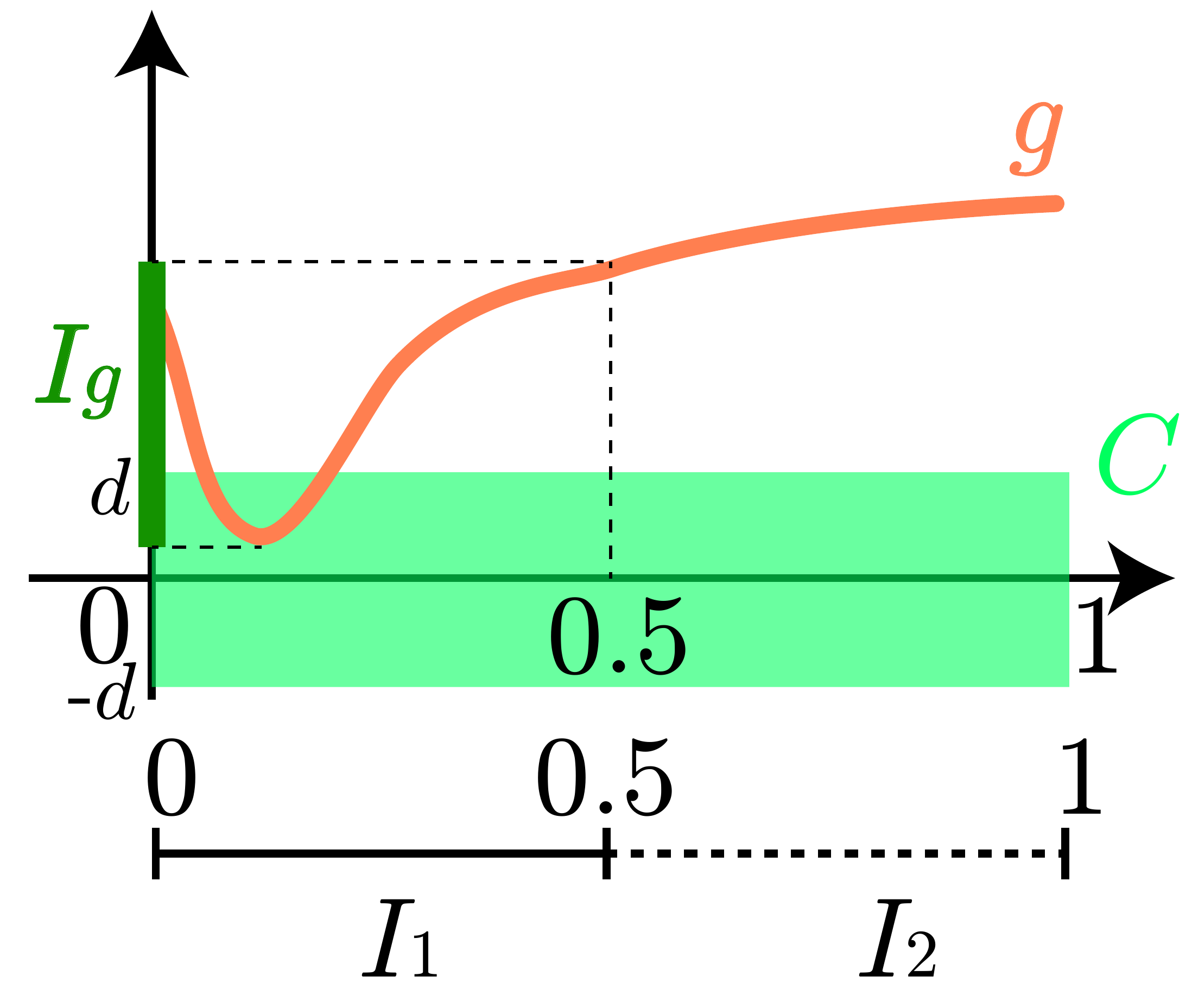}\hfill
    \includegraphics[width=.33\linewidth]{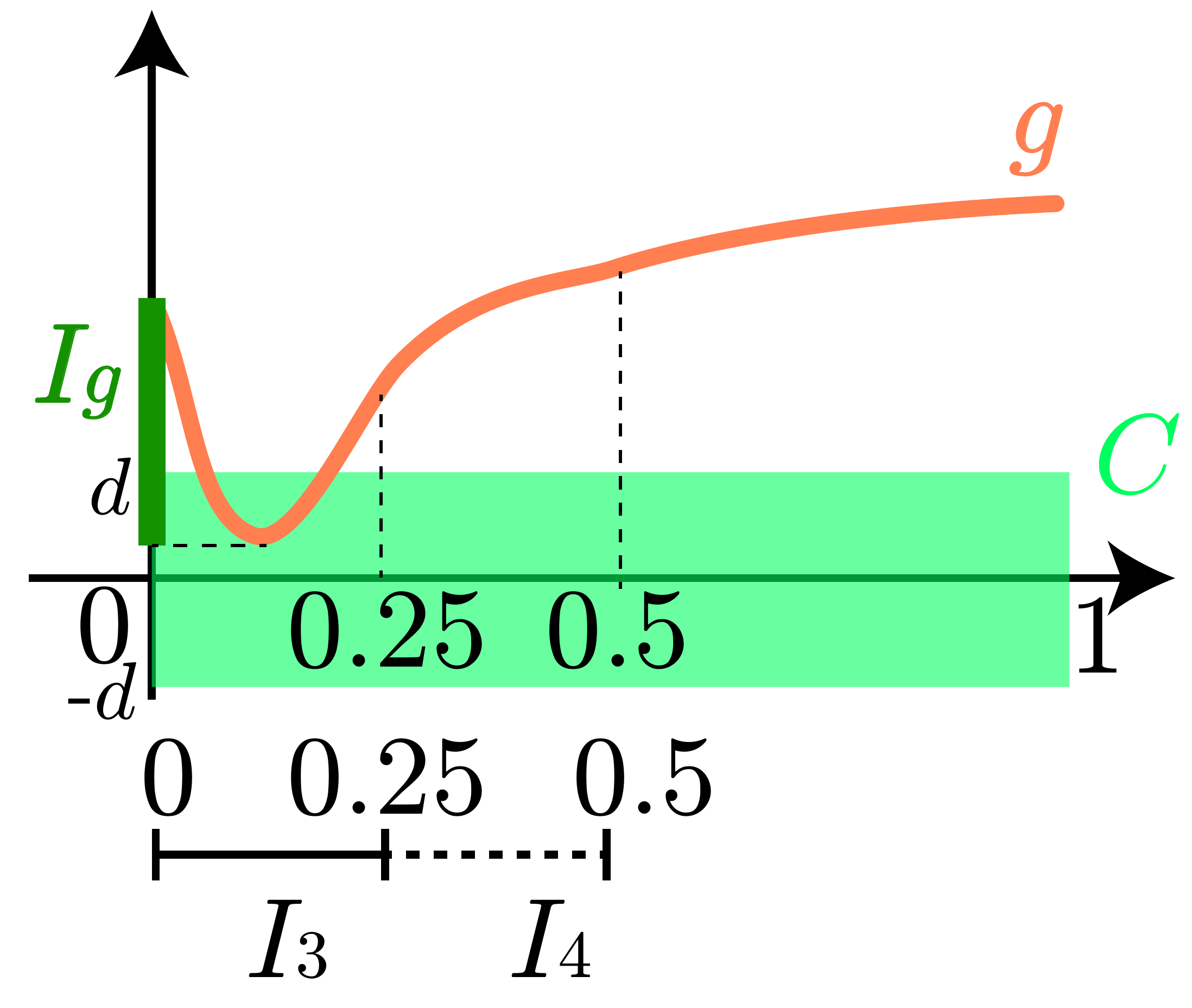}\par
    \parbox{.33\linewidth}{\centering\scriptsize{$\ell=0$}}\hfill
    \parbox{.33\linewidth}{\centering\scriptsize{$\ell=1$}}\hfill
    \parbox{.33\linewidth}{\centering\scriptsize{$\ell=2$}}\par
    \caption{\review{1D illustration of the first three levels of our MSCCD inclusion based root-finder. Instead of checking if $I_g$ intersects with the origin, we check if it intersects the interval $[-d, d]$ marked in light green.}}
    \label{fig:msccd-illustration}
\end{figure}

We wish to check whether the box $B_F(\Omega)$ intersects a cube of side $2d$ centered on the origin (Figure~\ref{fig:msccd-illustration}). Equivalently, we can construct another box $B_F'(\Omega)$ by displacing the six faces of $B_F(\Omega)$ outward at a distance $d$, and then check whether this enlarged box contains the origin.
This check can be done as for the standard CCD (Section~\ref{sec:method}), but the floating point filters must be recalculated to account for the additional sum (indeed, we add/subtract $d$ to/from all the coordinates). Hence, the filters for $F'$ are:
\begin{equation}\label{eq:msccd-error}
\begin{array}{c}
\epsilon_\text{ee}^x = 7.105427357601002 \times 10^{-15} \gamma_x^3 \\
\epsilon_\text{vf}^x = 7.549516567451064 \times 10^{-15} \gamma_x^3
\end{array}
\end{equation}
As before, the filters are calculated as described in~\cite{attene20} and they additionally assume that $d < \gamma_x$.

To account for minimum separations, the only change in our algorithm is at line~\ref{ll:compute_bf} where we need to enlarge $B$ by $d$ and in lines~\ref{ll:box-check} and~\ref{ll:termination-check} since $\C_\varepsilon$ needs to be replaced with $\C_\epsilon = [-\epsilon^x, \epsilon^x]\times [-\epsilon^y, \epsilon^y]\times[-\epsilon^z, \epsilon^z]$.


\subsection{Results}\label{sec:msccd-results}

To the best of our knowledge, the minimum separation floating-point time-of-impact root finder~\cite{harmon2011IAGM} (MSRF) implemented in~\cite{Libin:2019}, is the only public code supporting minimal separation queries. \review{While not explicitly constructed for MSCCD, FPRF uses a distance tolerance to limit false negatives, similarly to an explicit minimum separation. We compare the results and performance in Appendix~\ref{app:fprf-msccd}.}
\paragraph{\review{MSRF}} uses the univariate formulation, which requires to find the roots of a high-order polynomial, and it is thus unstable when implemented using floating-point arithmetic.

Table~\ref{tab:us-vs-stiv} reports timings, false positive, and false negatives for different separation distances $d$. 
As $d$ shrinks (around $10^{-16}$) the results of our method with MSCDD coincide with the ones with $d=0$ since the separation is small. For these small tolerances, MSRF runs into numerical problems and the number of false negatives increases. Figure~\ref{fig:eps-vs-time} shows the average query time versus the separation distance $d$ for the simulation dataset, since our method only requires to check the intersection between boxes, the running time largely depends on the number of detected collision, and the average is only mildly affected by the choice of $d$.

\begin{figure}
    \centering
    \parbox{0.02\linewidth}{~}\hfill
    \parbox{.48\linewidth}{\centering Vertex-Face MSCCD}\hfill
    \parbox{.48\linewidth}{\centering Edge-Edge MSCCD}\par
    \parbox{0.02\linewidth}{\centering\rotatebox{90}{\scriptsize{Average Running Time ($\mu$s)}}}\hfill
    \parbox{.48\linewidth}{\centering\includegraphics[width=0.9\linewidth]{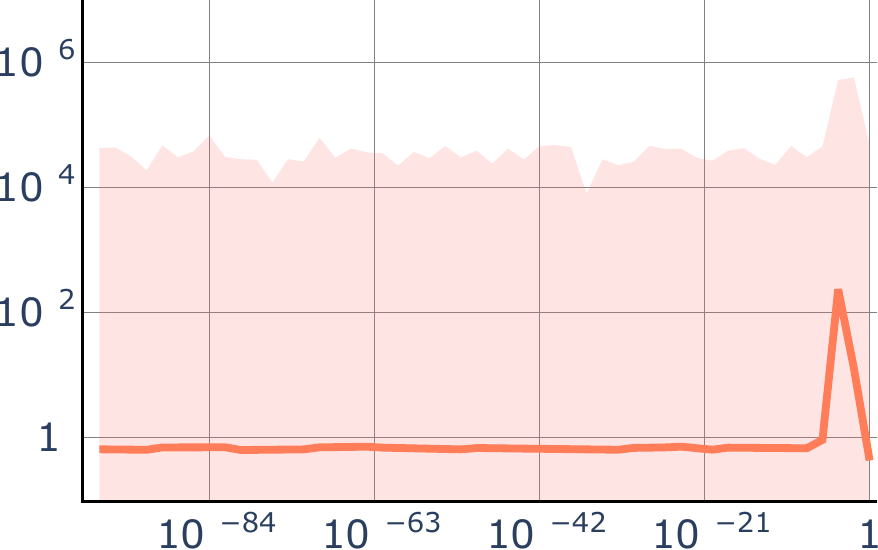}}\hfill
    \parbox{.48\linewidth}{\centering\includegraphics[width=0.9\linewidth]{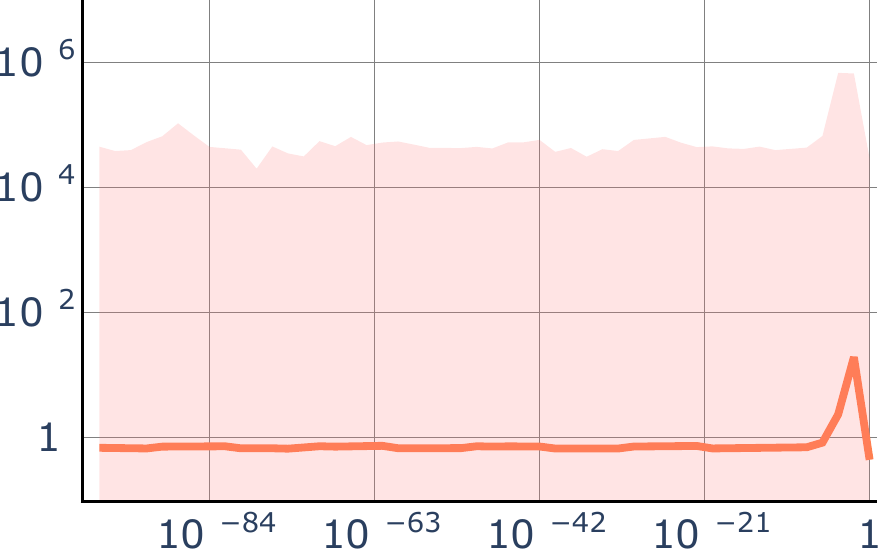}}\par
    \parbox{0.02\linewidth}{~}\hfill
    \parbox{.48\linewidth}{\centering\scriptsize{Distance $d$}}\hfill
    \parbox{.48\linewidth}{\centering\scriptsize{Distance $d$}}\\[2ex]
    \parbox{0.02\linewidth}{\centering\rotatebox{90}{\scriptsize{\# Queries}}}\hfill
    \parbox{.48\linewidth}{\centering\includegraphics[width=0.9\linewidth]{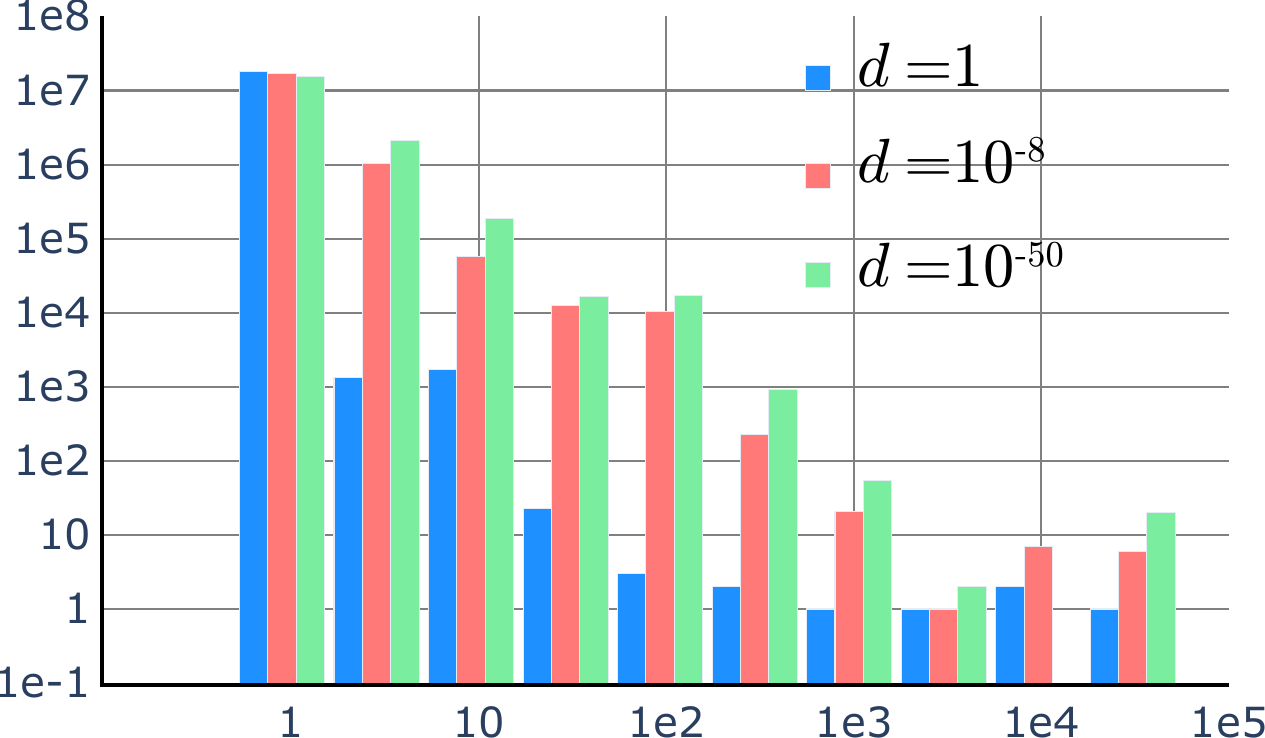}}\hfill
    \parbox{.48\linewidth}{\centering\includegraphics[width=0.9\linewidth]{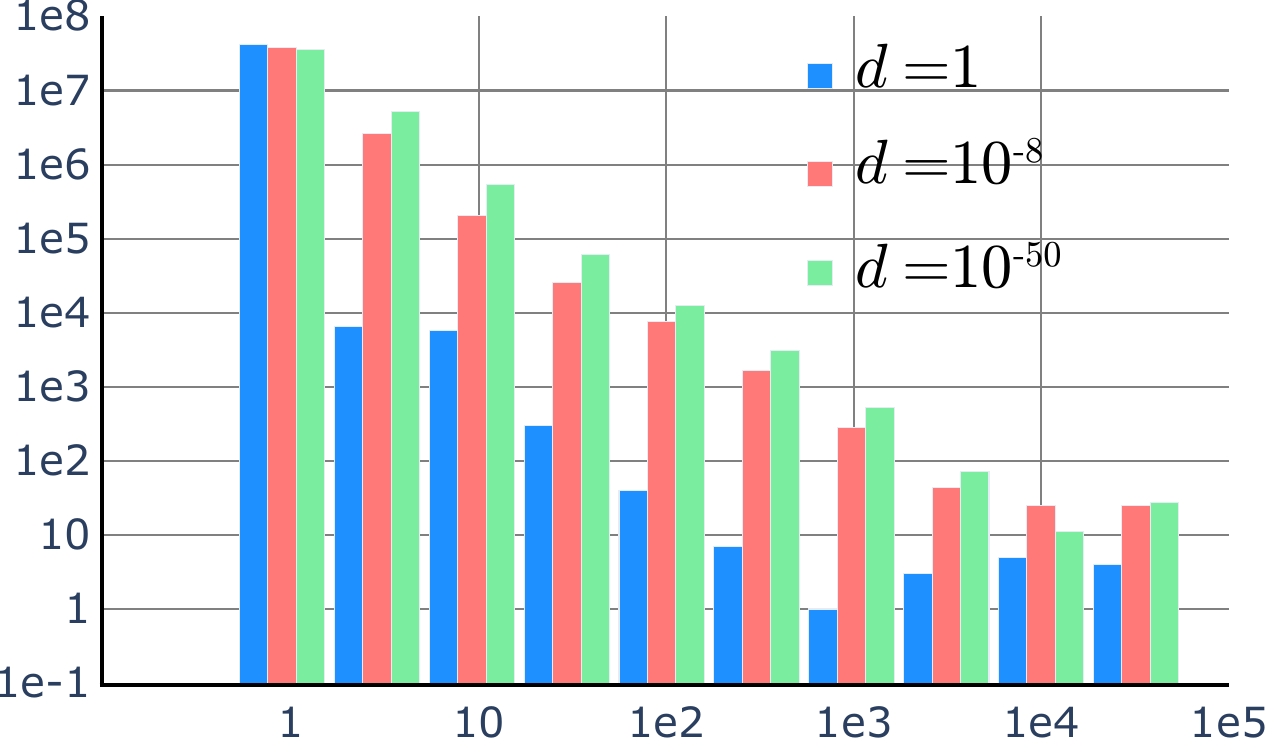}}\par
    \parbox{0.02\linewidth}{~}\hfill
    \parbox{.48\linewidth}{\centering\scriptsize{Running Time ($\mu$s)}}\hfill
    \parbox{.48\linewidth}{\centering\scriptsize{Running Time ($\mu$s)}}\par
    \caption{\review{Top, average runtime of our algorithm for varying minimum separation $d$ in the simulation dataset. The shaded area depicts the range of the values.} \revieww{Bottom, distribution of running time for three different minimum separation distanced $d = 10^{-50}$, $10^{-8}$, and $1$ over the simulation dataset.}}
    \label{fig:eps-vs-time}
\end{figure}

%% file: figures/msccd-table.tex
\begin{table*}
    \caption{{Summary of the average runtime in ${\mu}s$ (t), number of false positive (FP), and number of false negative (FN) for MSRF and our method.}}
    \centering
    \scriptsize
\setlength\tabcolsep{4pt}
\begin{tabular}{lccc|ccc|ccc|ccc|ccc|ccc|ccc|ccc}
\toprule
{} & \multicolumn{6}{|c|}{\Handcrafted{} -- Vertex-Face MSCCD} & \multicolumn{6}{c|}{\Handcrafted{} -- Edge-Edge MSCCD} & \multicolumn{6}{c|}{Simulation -- Vertex-Face MSCCD} & \multicolumn{6}{c}{Simulation -- Edge-Edge MSCCD} \\
\midrule
\multicolumn{1}{l|}{} & \multicolumn{3}{c|}{MSRF} & \multicolumn{3}{c|}{Ours} & \multicolumn{3}{c|}{MSRF} & \multicolumn{3}{c|}{Ours} & \multicolumn{3}{c|}{MSRF} & \multicolumn{3}{c|}{Ours} & \multicolumn{3}{c|}{MSRF} & \multicolumn{3}{c}{Ours} \\
\midrule
\multicolumn{1}{l|}{$d$} &    t &   FP &   FN &    t &    FP & FN &    t &   FP &   FN &     t &    FP & FN &     t &      FP &     FN &    t &       FP & FN &     t &      FP &    FN &    t &       FP & FN \\
\midrule
\multicolumn{1}{l|}{$10^{-2}$}    & 12.89 & 854 & 114 & 18.86K & 2.6K & 0 & 3.84 & 774 & 189 & 9.64K & 4.8K & 0 & 55.47 & 156.8K & 18.3K & 12.04 & 8.1M & 0 & 14.42 & 354.1K & 7.0K & 19.12 & 8.3M & 0\\
\multicolumn{1}{l|}{$10^{-8}$}   & 15.05 & 216 & 2 & 1.60K & 159 & 0 & 2.89 & 230 & 18 & 3.42K & 309 & 0 & 55.26 & 75 & 0 & 0.72 & 8 & 0 & 11.12 & 228 & 1 & 0.73 & 40 & 0\\
\multicolumn{1}{l|}{$10^{-16}$}   & 13.90 & 151 & 35 & 1.51K & 108 & 0 & 2.90 & 231 & 21 & 2.92K & 214 & 0 & 54.83 & 4 & 3.8K & 0.71 & 2 & 0 & 10.70 & 10 & 4 & 0.72 & 17 & 0\\

\multicolumn{1}{l|}{$10^{-30}$}    & 13.59 & 87 & 141 & 1.39K & 108 & 0 & 2.89 & 118 & 157 & 2.79K & 214 & 0 & 53.73 & 0 & 10.2K & 0.66 & 2 & 0 & 10.68 & 0 & 1.7K & 0.67 & 17 & 0\\
\multicolumn{1}{l|}{$10^{-100}$}   & 14.45 & 16 & 384 & 1.43K & 108 & 0 & 3.05 & 14 & 335 & 2.82K & 214 & 0 & 53.53 & 0 & 18.6K & 0.66 & 2 & 0 & 10.59 & 0 & 5.0K & 0.68 & 17 & 0\\
\bottomrule
\end{tabular}
    \label{tab:us-vs-stiv}
\end{table*}

%% file: 05-application.tex
\section{Integration in Existing Simulators}\label{sec:applications}



\review{In a typical simulation the objects are represented using triangular meshes and the vertices are moving along a linear trajectory in a timestep. At each  timestep, collisions might happen when a vertex hits a triangle, or when an edge hits another edge. A CCD algorithm is then used to prevent interpenetration; this can be done in different ways. In an active set construction method (Section~\ref{sec:active-set}) the CCD is used to compute contact forces to avoid penetration assuming linearized contact behaviour. For a line-search based method (Section~\ref{sec:line-search}), CCD and time of impact are used to prevent the Newton trajectory from causing penetration by limiting the step length. Note that, the latter approach requires a conservative CCD, while the former can tolerate false negatives.}

\review{The integration of a CCD algorithm with collision response algorithms is a challenging problem on its own, which is beyond the scope of this paper. As a preliminary study, to show that our method can be integrated in existing response algorithm, }
we examine two use cases in elastodynamic simulations:
\begin{enumerate}
    \item constructing an active set of collision constraints~\cite{Wriggers1995Finte, verschoor2019efficient, harmon2008robust}, Section~\ref{sec:active-set};
    \item during a line search to prevent intersections~\cite{Li2020IPC}, Section~\ref{sec:line-search}.
\end{enumerate}

\review{We leave as future work a more comprehensive study including how to use our CCD to further improve the physical fidelity of existing simulators or how to deal with challenging cases such as sliding contact response.}

To keep consistency across queries, we compute the numerical tolerances~\eqref{eq:ccd-error} and~\eqref{eq:msccd-error} for the whole scene. That is, $x_{\max}$, $y_{\max}$, and $z_{\max}$ are computed as the maximum over all the vertices in the simulation. \review{In algorithms \ref{alg:active-set} and \ref{alg:line-search} we utilize a broad phase method (e.g., spatial hash) to reduce the number of candidates $C$ that need to be evaluated with out narrow phase CCD algorithm.} 

\subsection{Active Set Construction}\label{sec:active-set}

\begin{algorithm}
    \begin{algorithmic}[1]
    \Function{ConstructActiveSet}{$x_0, x_1, \delta, m_I$}\label{alg:line:activeset}
    \State $C \gets$ \Call{\review{BroadPhase}}{$x_0, x_1$}
    \State $C_A \gets \emptyset$
    \For{$c \in C$}\Comment{\review{Iterate over the collision candidates}}
        \State $t \gets$ \Call{CCD}{$x_0 \cap c, x_1 \cap c, \delta, m_I$}
        \If{$0 \leq t \leq 1$}
            \State $C_A \gets C_A \cup \{(c, t)\}$
        \EndIf
    \EndFor
    \State \Return $C_A$
    \EndFunction
    \State
    \revieww{
    \Function{CCD}{$c_0, c_1, \delta, m_I$}
    \If{$c_0$ and $c_1$ are edges}
        \State $F\gets$ build $F_{\text{ee}}$ from $c_0$ and $c_1$\Comment{Equation~\eqref{eq:F-ee}}
    \Else
        \State $F\gets$ build $F_{\text{vf}}$ from $c_0$ and $c_1$\Comment{Equation~\eqref{eq:F-vf}}
    \EndIf
    \State \Return \Call{Solve}{$F, \delta, m_I$}
    \EndFunction
    }
    \end{algorithmic}
    
    \caption{Active Set Construction Using Exact CCD}
    \label{alg:active-set}
\end{algorithm}

\review{In the traditional constraint based collision handling (such as that of~\citet{verschoor2019efficient}), collision response is handled by performing an implicit timestep as a constrained optimization. The goal is to minimize a elastic potential while avoiding interpenetration through gap constraints. } To avoid handling all possible collisions during a simulation, a subset of  \textit{active} collisions constraints \review{$C_A$} is usually constructed. This set not only avoids infeasibilities, but also improves performance by having fewer constraints. There are many activation strategies, but for the sake of brevity we focus here on the strategies used by~\citet{verschoor2019efficient}. 

\review{Algorithm \ref{alg:active-set} shows how CCD is used to compute the active set $C_A$. Given the starting and ending vertex positions, $x_0$ and $x_1$, we compute the time of impact for each collision candidate $c \in C$. We use the notation $x_i \cap c$ to indicate selecting the constrained vertices from $x_i$. If the candidate $c$ is an actual collision, that is $0 \leq t \leq 1$, then we add this constraint and the time of impact, $t$, to the active set, $C_A$.}

\review{From the active constraint set the constraints of~\citet{verschoor2019efficient} are computed as}
$$
\langle {n}, p_{c}^1 - \p_{c}^2 \rangle \geq 0,
$$
where ${n}$ is the contact normal \review{(i.e., for a point-triangle the triangle normal at the time of impact and for edge-edge the edge-edge cross product at the time of impact)}, $p_{c}^1$ is the point (or the contact point on the first edge), and $p_{c}^2$ is the point of contact on the triangle (or on the second edge) at the end of the timestep. 
Note that, this constraint requires to compute the point of contact, which depends on the the time-of-impact which can be obtained directly from our method.

Because of the difficulty for a simulation solver to maintain and not violate constraints, it is common to offset the constraints such that 
\[
\langle {n}, p_{c}^1 - p_{c}^2 \rangle \geq \eta > 0.
\]
In such a way, even if the $\eta$ constraint is violated, the real constraint is still satisfied. This common trick, implies that the constraints need to be activated early (i.e., when the distance between two objects is smaller than $\eta$) which is exactly what our MSCCD can compute when using $d = \eta$. In Figure~\ref{fig:active-set-delta}, we use a value of $\eta=\SI{0.001}{\meter}$. When using large values of $\eta$, the constraint of~\citet{verschoor2019efficient} can lead to infeasibilities because all triangles are extended to planes and edges to lines.

Figure~\ref{fig:active-set-delta} shows example of simulations run with different numerical tolerance $\delta$. Changing $\delta$ has little effect on the simulation in terms of run-time, but for large values of $\delta$, it can affect accuracy. We observe that for a $\delta \geq 10^{-2}$ the simulation is more likely to contain intersections. This is most likely due to the inaccuracies in the contact points used in the constraints. 

\begin{figure}
    \centering
    \includegraphics[width=\linewidth]{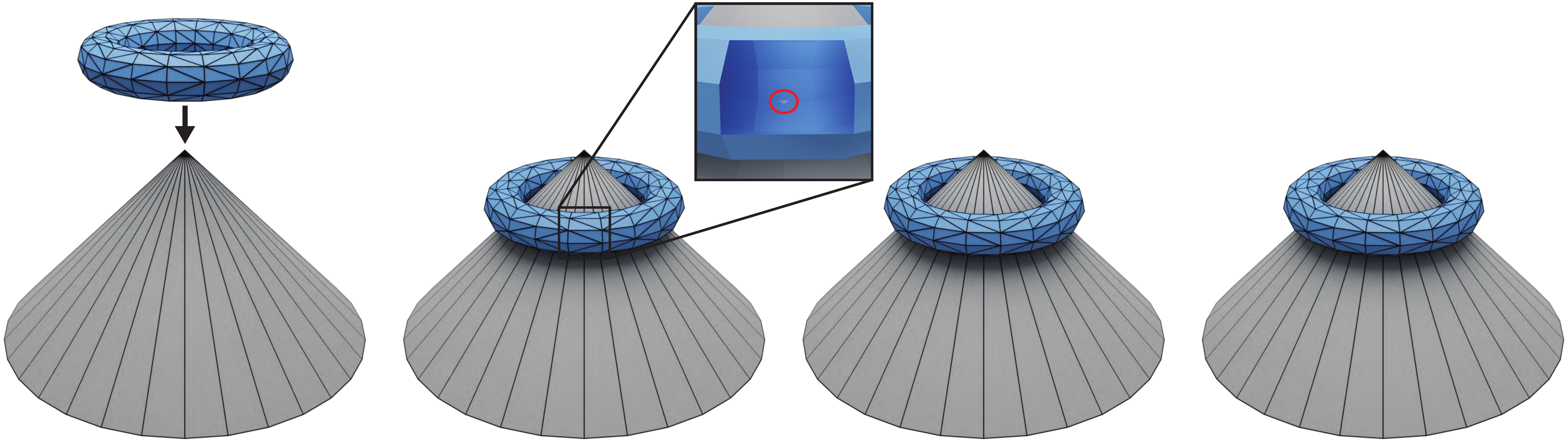}\\
    \hspace*{0.25\linewidth}
    \parbox[0]{0.24\linewidth}{\centering $\delta=10^{-1}$}
    \parbox[0]{0.24\linewidth}{\centering $\delta=10^{-3}$}
    \parbox[0]{0.24\linewidth}{\centering $\delta=10^{-6}$}
    \caption{An elastic simulation using the constraints and active set method of~\citet{verschoor2019efficient}. From an initial configuration (left) we simulate an elastic torus falling on a fixed cone using three values of $\delta$ (from left to right: $10^{-1}, 10^{-3}, 10^{-6}$). The total runtime of the simulation is affected little by the change in $\delta$ ($24.7$, $25.2$, and $26.2$ seconds from left to right compared to $32.3$ seconds when using FPRF). For $\delta=10^{-1}$, inaccuracies in the time-of-impact lead to inaccurate contact points in the constraints and, ultimately, intersections (inset). }
    \label{fig:active-set-delta}
\end{figure}

\subsection{Line Search}\label{sec:line-search}



    
    

\begin{algorithm}
    \begin{algorithmic}[1]
\Function{LineSearch}{$E, x_0, \Delta x, p, \delta, m_I$}

\State $x_1 \gets x_0 + \Delta{x}$
\State $C \gets $ \Call{\review{BroadPhase}}{$x_0, x_1$}\Comment{Collision candidates}
\State $\alpha \gets 1$
\State $d_i, \rho_i \gets \Call{Distance}{C}$\label{ll:distance-c}
\State Compute $\epsilon_i$ from~\eqref{eq:msccd-error}
\State $d \gets \max(p \d_i, \delta)$\label{ll:ls-msccd}
\While{$p  < (d - \delta - \epsilon_i - \rho_i)/d$}\label{ll:find-p}
    \State$p\gets p/2$
    \State $d \gets p \d_i$
\EndWhile
\State$\delta_i\gets\delta$
\For{$c \in C$}\Comment{$\alpha$ is bounded by earliest time-of-impact}\label{ll:msccd}
    \State $t, \bar\delta_i \gets$ \Call{MSCCD}{$x_0 \cap c, x_1 \cap c, d, \alpha, \delta, m_I$} \label{ll:msccd-ls}
    \State $\alpha \gets \min(t, \alpha)$
    \State $\delta_i \gets \max(\bar\delta_i, \delta_i)$
\EndFor
\If{$p  < (d - \delta_i - \epsilon_i - \rho_i)/d$}\label{ll:validate-after}
\State$\delta \gets \delta_i$\Comment{Repeat with $p$ validated from $\delta_i$}
\State Go to line~\ref{ll:find-p}.
\EndIf
\State
\While{$\alpha > \alpha_{\min}$}\Comment{\review{Backtracking} line-search}
    \State $x_1 \gets x_0 + \alpha \Delta{x}$
    \If{$E(x_1$) < $E(x_0)$ }\Comment{Objective energy decrease}\label{ll:energy-dec}
        \State \textbf{break}
    \EndIf
    
    \State $\alpha \gets \alpha/2$
    
\EndWhile
\State \Return $\alpha$
\EndFunction
\State
\revieww{
\Function{MSCCD}{$c_0, c_1, d, t, \delta, m_I$}
\If{$c_0$ and $c_1$ are edges}
    \State $F\gets$ build $F_{\text{ee}}$ from $c_0$ and $c_1$\Comment{Equation~\eqref{eq:F-ee}}
\Else
    \State $F\gets$ build $F_{\text{vf}}$ from $c_0$ and $c_1$\Comment{Equation~\eqref{eq:F-vf}}
\EndIf
\State \Return \Call{SolveMSCCD}{$F, t, \delta, m_I, d$}
\EndFunction
}
    
\end{algorithmic}
\caption{Line Search with Exact CCD}
\label{alg:line-search}
\end{algorithm}

\review{A line search is used in a optimization to ensure that every update decreases the energy $E$. That is, given an update, $\Delta x$, to the optimization variable $x$, we want to find a step size $\alpha$ such that $E(x + \alpha \Delta x) < E(x)$. This ensure that we make progress towards a minimum.}

When used in a line search algorithm, CCD can be used to prevent intersections and tunneling. This requires modifying the maximum step length to the time of impact. As observed by~\citet{Li2020IPC}, the standard CCD formulation without minimal separation cannot be used directly in a line search algorithm. Let $t^\star$ the earliest time of impact (i.e., $F(t^\star, \tilde u, \tilde v) = 0$ for some $\tilde u, \tilde v$ and there is no collision between $0$ and  $t^\star$) and assume that the energy at $E(x_0 + t^\star \Delta x) < E(x_0)$ (Algorithm~\ref{alg:line-search}, line~\ref{ll:energy-dec}). In this case the step $\alpha = t^\star$ is a valid descent step which will be used to update the position $x$ in outer iteration (e.g., Newton optimization loop). In the next iteration, the line search will be called with the updated position and the earliest time of impact will be zero since we selected $t^\star$ in the previous iteration. This prevents the optimization from making progress because any direction $\Delta x$ will lead to a time of impact $t=0$. To avoid this problem we need the line search to find an appropriate step-size $\alpha$ along the update direction that leaves ``sufficient space'' for the next iteration, so that the barrier in~\cite{Li2020IPC} will be active and steer the optimization away from the contact position. Formally, we aim at finding a \emph{valid CCD sequence} $\{t_i\}$ such that
\begin{align*}
    t_i < t_{i+1},\quad
    \lim_{i\to\infty} t_i = t^\star,\quad
    \text{and}\quad
    t_i/t_{i+1}\approx 1.
\end{align*}
The first requirement ensures that successive CCD checks will report an increasing time, the second one ensures that we will converge to the true minimum, and the last one aims at having a ``slowly'' convergent sequence (necessary for numerical stability). \citet{Li2020IPC} exploit a feature of FPRF to simulate a minimal separation CCD: in this work we propose to directly use our MSCCD algorithm (Section \ref{sec:msccd}).

\paragraph{Constructing a Sequence.} Let $0<p<1$ be a user-defined tolerance ($p$ close to 1 will produce a sequence $\{t_i\}$ converging faster) and $d_i$ be the distance between two primitives. We propose to set $d=p d_i$, and ensure that no primitive are closer than $d$. Without loss of generality, we assume that $F(x+\Delta x) =0$, that is, taking the full step will lead to contact. By taking successive steps in the same direction, $d_i$ will shrink to zero ensuring $t_i$ to converge to $t^\star$. Similarly we will obtain a growing sequence $t_i$ since $d$ decreases as we proceed with the iterations. Finally, it is easy to see that $p=t_i/t_{i+1}$ which can be close to one.

To account for the aforementioned problem, we propose to use our MSCCD algorithm to return a valid CCD sequence when employed in a line search scenario. For a step $i$, we define $\delta^i$ as the tolerance, $\epsilon_i$ the numerical error~\eqref{eq:msccd-error}, and $\rho_i$ as the maximum numerical error in computing \review{the distances} $d_i$ \review{from the candidates set $C$ (line~\ref{ll:distance-c})}. $\rho_i$ should be computed using forward error analysis on the application-specific distance computation: since the applications are not the focus of our paper, we used a fixed  $\rho_i=10^{-9}$, and we leave the complete forward analysis as a future work. (We note that our approximation might thus introduce zero length steps, this however did not happen in our experiments.)
If $d_i -(\delta_i +\epsilon_i + \rho_i) > d$, our MSCCD is guaranteed to find a time of impact larger than zero. Thus if we set $d=p d_i$ (line~\ref{ll:ls-msccd}), we are guaranteed to find a positive time of impact if
\[
d_i > \frac{\delta_i + \epsilon_i + \rho_i}{1-p}.
\]
To ensure that this inequality holds, we propose to validate $p$ before using the MSCCD with $\delta$ (line~\ref{ll:find-p}), find the time of impact and the actual $\delta_i$ (line~\ref{ll:msccd}), and check if the used $p$ is valid (line~\ref{ll:validate-after}). In case $p$ is too large, we divide it by two until it is small enough. Note that, it might be that
\[
d_i < \delta_i + \epsilon_i + \rho_i,
\]
in this case we can still enforce the inequality by increasing the number of iterations, decreasing $\delta$, or using multi-precision in the MSCCD to reduce $\epsilon_i$. However, this was never necessary in any of our queries, and we thus leave a proper study of these options as a future work.

As visible from Table~\ref{tab:us-vs-stiv}, our MSCCD slows down as $d$ grows. Since the actual minimum distance is not relevant in the line search algorithm, our experiments suggest to cap it at $\delta$ (line~\ref{ll:ls-msccd}). To avoid unnecessary computations and speedup the MSCCD computations, our algorithm, as suggested by~\citet{Redon2002fast}, can be easily modified to accept a shorter time interval (line~\ref{ll:msccd-ls}): it only requires to  change the initialization of $I$ \revieww{(Algorithm~\ref{alg:our} line~\ref{ll:init-ours})}. These two modifications lead to a $8\times$ 
speedup in our experiments. \revieww{We refer to this algorithm with MSCCD (i.e., Algorithm~\ref{alg:our} with MSCDD, Section~\ref{sec:msccd-method}, and modified initialization of $I$) as \texttt{SolveMSCCD}.}


Figure~\ref{fig:line-search-delta} shows a simulation using our MSCCD in line search to keep the bodies from intersecting for different $\delta$. As illustrated in the previous section, the effect of $\delta$ is negligible as long as $\delta\leq 10^{-3}$. Timings vary depending on the maximum number of iterations. Because the distance $d$ varies throughout the simulation, some steps take longer than others (as seen in Figure~\ref{fig:eps-vs-time}). We note that, if we use the standard CCD formulation $F=0$, the line search gets stuck in all our experiments, and we were not able to find a solution. \review{Note that for a line search based method it is crucial to have a conservative CCD/MSCCD algorithm: the videos in the additional material shows that a false negative leads to an artefact in the simulation.}

\begin{figure}
    \centering
    \includegraphics[width=\linewidth]{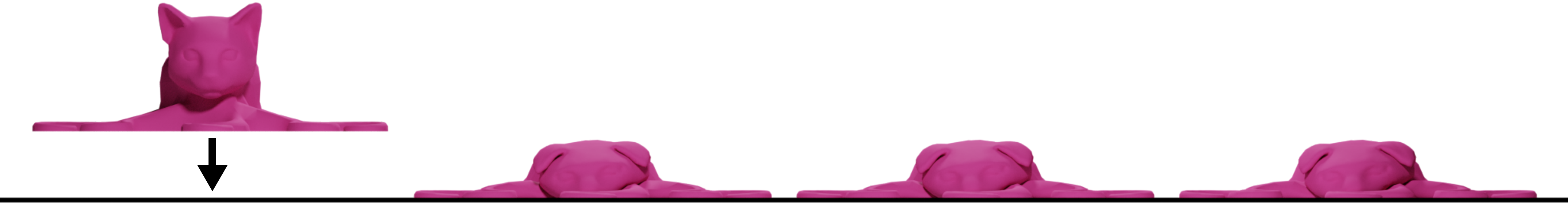}\\
    \hspace*{0.25\linewidth}
    \parbox[0]{0.24\linewidth}{\centering $\delta=10^{-3}$}
    \parbox[0]{0.24\linewidth}{\centering $\delta=10^{-4.5}$}
    \parbox[0]{0.24\linewidth}{\centering $\delta=10^{-6}$}
    \caption{An example of an elastic simulation using our line search (Section~\ref{sec:line-search}) and the method of~\citet{Li2020IPC} to keep the bodies from intersecting. An octocat is falling under gravity onto a triangulated plane. From left to right: the initial configuration, the final frame with $\delta=10^{-3}$, $\delta=10^{-4.5}$, $\delta=10^{-6}$ all with a maximum of $10^{6}$ iterations. There are no noticeable differences in the results, and the entire simulations takes 
    $63.3$, $67.9$, and $67.0$ 
    seconds from left to right (a speed up compared to using FPRF which takes $102$ seconds). \copyright Brian Enigma under CC BY-SA 3.0.
    }
    \label{fig:line-search-delta}
\end{figure}






%% file: 06-conclusion.tex
\section{Limitations and Concluding Remarks}

We constructed a benchmark of CCD queries and used it to study the properties of existing CCD algorithms. The study highlighted that the multivariate formulation is more amenable to robust implementations, as it avoids a challenging filtering of spurious roots. This formulation, paired with an interval root finder and modern predicate construction techniques leads to a novel simple, robust, and efficient algorithm, supporting minimal separation queries with runtime comparable to state of the art, non conservative, methods. 

While we believe that it is practically acceptable, our algorithm still suffers from false positive and it will be interesting to see if the multivariate root finding could be done exactly with reasonable performances, for example employing expansion arithmetic in the predicates. Our definition of minimal separation distance is slightly different from the classical definition, and it would be interesting to study how to extend out method to directly support Euclidean distances. Another interesting venue for future work is the extension of our inclusion function to non-linear trajectories \review{and their efficient evaluation using static filters or exact arithmetic.}

\reviewww{Our benchmark focuses only on CPU implementations: reimplementing our algorithm on a GPU with our current guarantees is a major challenge. It will require to control the floating-point rounding on the GPU (and compliant with the IEEE floating-point standard), to ensure that the compiler does not reorder the operations or skip the computation of temporaries. Additionally it would require to recompute the ground truth and the numerical constants for single precision arithmetic, as most GPUs do not yet support double computation. This is an exciting direction for future work to further improve the performance of our approach.}

We will release an open-source reference implementation of our technique with an MIT license to foster adoption of our technique by existing commercial and academic simulators. We will also release the dataset and the code for all the algorithms in our benchmark to allow researchers working on CCD to easily compare the performance and correctness of future CCD algorithms.

%% file: 07-acknowledgement.tex
\ifx\archive\undefined
\begin{acks}
\else
\section{Acknowledgements}
\fi
We thank Danny Kaufman for valuable discussions and NYU IT High Performance Computing for resources, services, and staff expertise.
This work was partially supported by the NSF CAREER award under Grant No. 1652515, the NSF grants OAC-1835712, OIA-1937043, CHS-1908767, CHS-1901091, National Key Research and Development Program of China No. 2020YFA0713700, EU ERC Advanced Grant CHANGE No. 694515, a Sloan Fellowship, a gift from Adobe Research, a gift from nTopology, and a gift from Advanced Micro Devices, Inc.

\ifx\archive\undefined
\end{acks}
\fi

%% file: 10-appendix.tex
\input{figures/msccd-fprf}

\section{Dataset Format}\label{app:format}

To avoid any loss of precision we convert every input floating-point coordinate in rationals using GMP~\cite{gmp}. This conversion is exact since every floating point can be converted in a rational number, as long as the numerator and denominator are arbitrarily large integers. We then store the numerator and denominator as a \texttt{string} since the numerator and denominator can be larger than a \texttt{long} number. To retrieve the floating point number we allocate a GMP rational number with the two strings and convert it to \texttt{double}.

In summary, one CCD query is represented by a $8 \times 7$ matrix where every row is one of the 8 CCD input points, and the columns are the interleaved $x, y, z$ coordinates of the point, represented as numerator and denominator. For convenience, we appended several such matrices in a common CSV file. The last column represents the result of the ground truth. For instance a CC query between $\p_1^0, \p_2^0, \p_3^0, \p_4^0$ and $\p_1^1, \p_2^1, \p_3^1, \p_4^1$ is represented as
{\footnotesize
\[
\begin{matrix}
p_{1^x_n}^0& p_{1^x_d}^0& p_{1^y_n}^0& p_{1^y_d}^0& p_{1^z_n}^0& p_{1^z_d}^0 & T\\
p_{2^x_n}^0& p_{2^x_d}^0& p_{2^y_n}^0& p_{2^y_d}^0& p_{2^z_n}^0& p_{2^z_d}^0 & T\\
p_{3^x_n}^0& p_{3^x_d}^0& p_{3^y_n}^0& p_{3^y_d}^0& p_{3^z_n}^0& p_{3^z_d}^0 & T\\
p_{4^x_n}^0& p_{4^x_d}^0& p_{4^y_n}^0& p_{4^y_d}^0& p_{4^z_n}^0& p_{4^z_d}^0 & T\\
p_{1^x_n}^1& p_{1^x_d}^1& p_{1^y_n}^1& p_{1^y_d}^1& p_{1^z_n}^1& p_{1^z_d}^1 & T\\
p_{2^x_n}^1& p_{2^x_d}^1& p_{2^y_n}^1& p_{2^y_d}^1& p_{2^z_n}^1& p_{2^z_d}^1 & T\\
p_{3^x_n}^1& p_{3^x_d}^1& p_{3^y_n}^1& p_{3^y_d}^1& p_{3^z_n}^1& p_{3^z_d}^1 & T\\
p_{4^x_n}^1& p_{4^x_d}^1& p_{4^y_n}^1& p_{4^y_d}^1& p_{4^z_n}^1& p_{4^z_d}^1 & T,
\end{matrix}
\]
}
where $p_{i^x_n}^t$ and $p_{i^x_d}^t$ are respectively the numerator and denominator of the $x$-coordinate of $p$, and $T$ is the same ground truth. \review{The dataset and a query viewer can be downloaded from the \href{https://archive.nyu.edu/handle/2451/61518}{NYU Faculty Digital Archive}.}

\section[Example of Degenerate Case not Properly Handled by Brochu et al. 2012]{Example of Degenerate Case not Properly Handled by~\cite{brochu2012efficient}}\label{app:britson-bug}
Let 
{\footnotesize{
\begin{equation}\label{eq:butterfly}
\begin{aligned}
\p^0 &= [0.1, 0.1, 0.1],&
\v_1^0 &= [0, 0, 1],&
\v_2^0 &= [1, 0, 1],&
\v_3^0 &= [0, 1, 1],&\\
\p^1 &= [0.1, 0.1, 0.1],&
\v_1^1 &= [0, 0, 0],&
\v_2^1 &= [0, 1, 0],&
\v_3^1 &= [1, 0, 0]&
\end{aligned}
\end{equation}
}}
be the input point and triangle. Checking if the point intersects the triangle is equivalent to check if the prism shown in Figure ~\ref{fig:butterfly} contains the origin. However, the prism contains a bilinear face that is degenerate (it looks like a ``hourglass''). The algorithm proposed in \cite{brochu2012efficient} does not consider this degenerate case and erroneously reports no collision.

\begin{figure}
    \centering
    \includegraphics[width=0.8\linewidth]{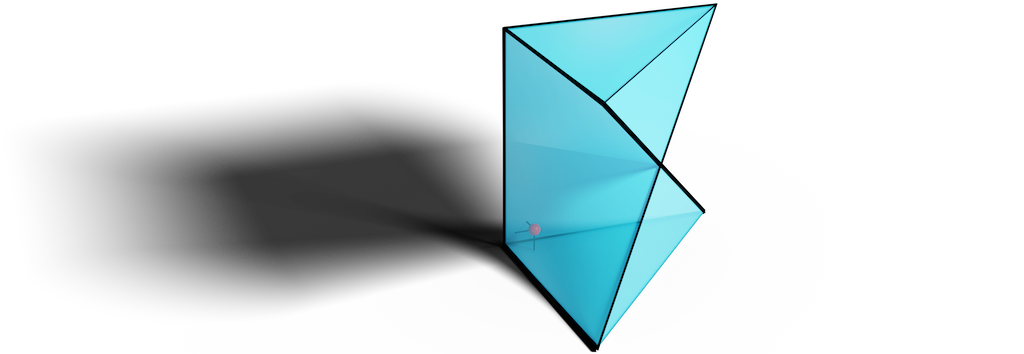}
    \caption{Prism resulting from the input points and triangle in~\eqref{eq:butterfly}. The origin is marked by the red dot.}
    \label{fig:butterfly}
\end{figure}

\section[Example of Inflection Point not Properly Handled by Tang et al. 2014]{Example of Inflection Point not Properly Handled by~\cite{tang2014fast}}\label{app:dinesh-bug}

Let
{\footnotesize{
\begin{align*}
\p^0 &= [1, 1, 0],&
\v_1^0 &= [0, 0, 5],&
\v_2^0 &= [2, 0, 2],&
\v_3^0 &= [0, 1, 0],&\\
\p^1 &= [1, 1, 0],&
\v_1^1 &= [0, 0, -1],&
\v_2^1 &= [0, 0, -2],&
\v_3^1 &= [0, 7, 0]&
\end{align*}
}}
be the input point and triangle. Checking if they intersect at time $t$ is equivalent to finding the roots of
\[
-72 t^3 + 120 t^2 - 44 t + 3.
\]
To apply the method in~\cite{tang2014fast} we need to rewrite the polynomial in form of~\citet[Equation (1)]{tang2014fast}:
\[
1 B_0^3(t)-\frac{35} 3 B_1^3(t) + \frac{82} 3  B_2^3(t) + 14  B_3^3(t).
\]
Their algorithm assumes no inflection points in the Bezier curve. Thus it proposes to split the curve at the eventual inflection point (as in the case above). The formula proposed in~\cite[Section 4.1]{tang2014fast} contains a typo, by fixing it we obtain the inflection point at:
\[
t = \frac{6k_0 - 4k_1 + k_2}{6k_0 - 6k_1 + 3k_2 - k_3} = \frac{5}{9}.
\]
By using the incorrect formula we obtain $t=155/312$, which  is not an inflection point. In both cases, $t$ cannot be computed exactly since it contains a division, and computing it approximately breaks the assumption of not having inflection  points in the Bezier form. In the reference code, the authors detect the presence of an inflection point using predicates, but do not split the curve (the case is not handled). We modified the code (patch attached in the additional material) to conservatively return a collision in these cases.

Independently from this problem, their reference implementation returns false negative (i.e. misses collisions) for certain configurations, such as the following degenerate configuration:

{\footnotesize{
\begin{align*}
\p^0 &= [1, 0.5, 1],&
\v_1^0 &= [0, 0.57, 1],&
\v_2^0 &= [1, 0.57, 1],&
\v_3^0 &= [1, 1.57, 1],&\\
\p^1 &= [1, 0.5, 1],&
\v_1^1 &= [0, 0.28, 1],&
\v_2^1 &= [1, 0.28, 1],&
\v_3^1 &= [1, 1.28, 1].&
\end{align*}
}}

We could not find out why this is happening, and we do not know if this is a theoretical or numerical problem, or a bug in the implementation.

\section[Effect of delta on the interval-based methods]{Effect of $\delta$ on the interval-based methods}\label{app:ib-params}
\review{UIRF, IRF, and our method have a single parameter $\delta$ to control the size of the interval. Increasing $\delta$ will introduce more false positive, while making the algorithms faster (Figure~\ref{fig:ib-params}). Note that we limit the total running time to 24h, thus UIRF does not have result for  $\delta > 10^{-6}$ (for $\delta = 10^{-6}$ it takes 1ms per query in average). 
$\delta$ has a similar effect on the number of false positives for the three interval based methods, while it has a more significant impact on the running time for UIRF and IRF.}

\begin{figure}
    \centering
    \parbox{0.02\linewidth}{~}\hfill
    \parbox{.48\linewidth}{\centering Vertex-Face CCD}\hfill
    \parbox{.48\linewidth}{\centering Edge-Edge CCD}\par
    \parbox{0.02\linewidth}{\centering\rotatebox{90}{\scriptsize{Average Runtime ($\mu$s)}}}\hfill
    \parbox{.48\linewidth}{\centering\includegraphics[width=0.9\linewidth]{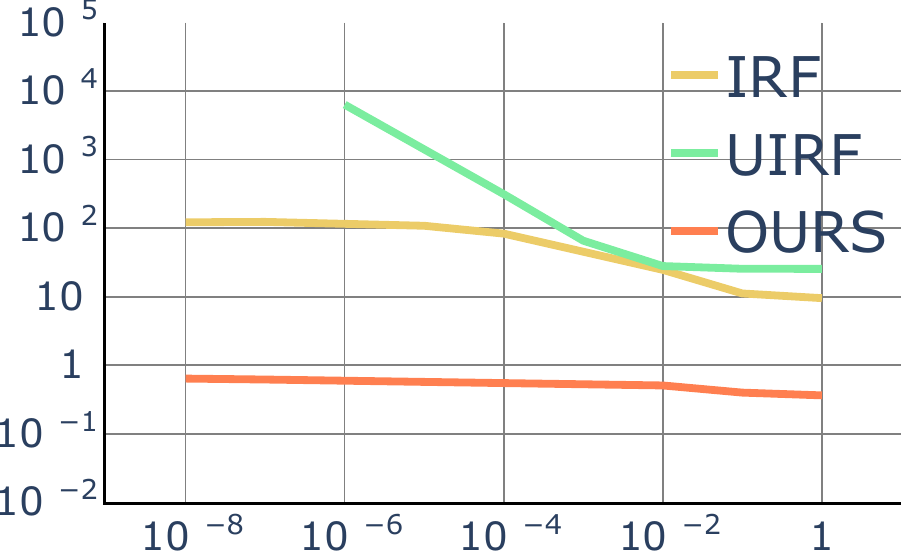}}\hfill
    \parbox{.48\linewidth}{\centering\includegraphics[width=0.9\linewidth]{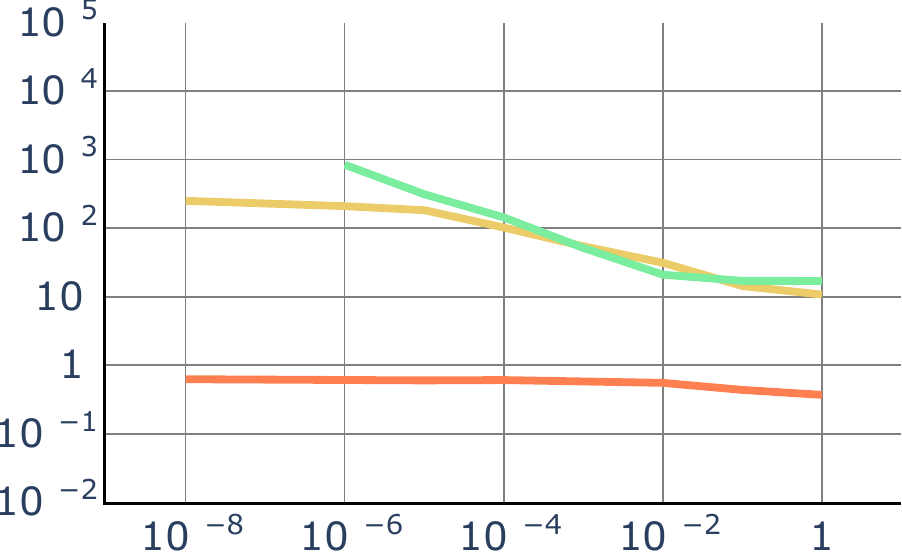}}\par
    \parbox{0.02\linewidth}{\centering\rotatebox{90}{\scriptsize{Number of False Positive}}}\hfill
    \parbox{.48\linewidth}{\centering\includegraphics[width=0.9\linewidth]{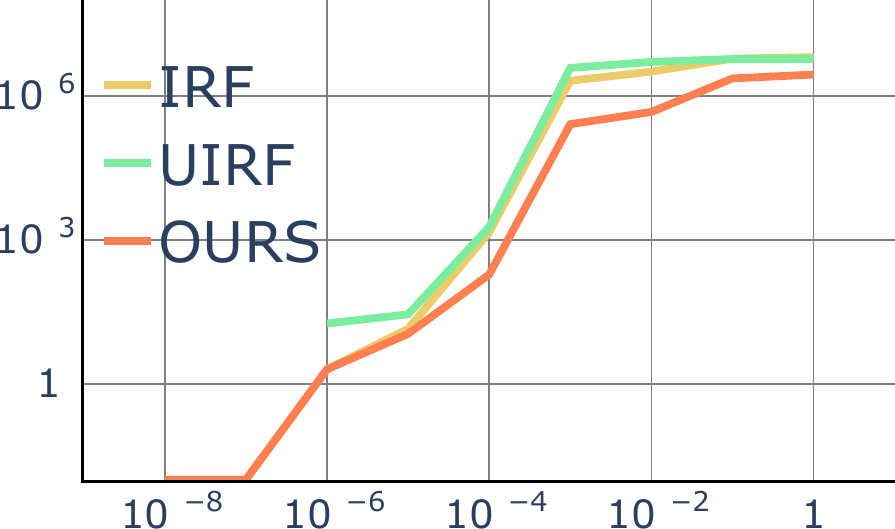}}\hfill
    \parbox{.48\linewidth}{\centering\includegraphics[width=0.9\linewidth]{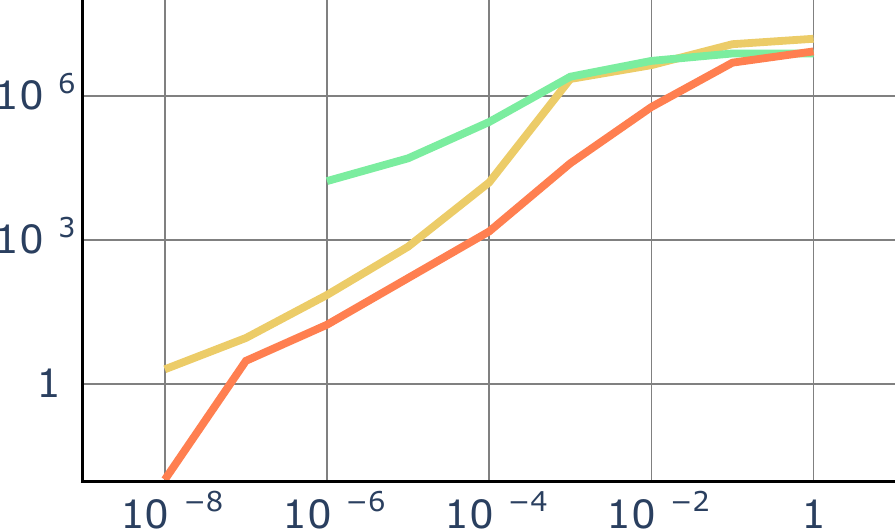}}\par
    \parbox{0.02\linewidth}{~}\hfill
    \parbox{.48\linewidth}{\centering\scriptsize{Tolerance $\delta$}}\hfill
    \parbox{.48\linewidth}{\centering\scriptsize{Tolerance $\delta$}}\par
    \caption{\review{Log plot of the effect of the tolerance $\delta$ on the running time (top) and false positives (bottom) for the three (Ours, UIRF, and IRF) interval based methods on the simulation dataset.}}
    \label{fig:ib-params}
\end{figure}

\section{Minimum separation with FPRF}\label{app:fprf-msccd}

\review{In Table~\ref{tab:us-vs-fprf} we compare our method with FPRF by changing the parameter $\eta$ that mimics minimum separation.}

%% file: figures/msccd-fprf.tex
\begin{table*}
    \caption{\review{Summary of the average runtime in ${\mu}s$ (t), number of false positive (FP), and number of false negative (FN) for FPRF and our method.}}
    \centering
    \scriptsize
\setlength\tabcolsep{4pt}
\begin{tabular}{lccc|ccc|ccc|ccc|ccc|ccc|ccc|ccc}
\toprule
{} & \multicolumn{6}{|c|}{\Handcrafted{} -- Vertex-Face MSCCD} & \multicolumn{6}{c|}{\Handcrafted{} -- Edge-Edge MSCCD} & \multicolumn{6}{c|}{Simulation -- Vertex-Face MSCCD} & \multicolumn{6}{c}{Simulation -- Edge-Edge MSCCD} \\
\midrule
\multicolumn{1}{l|}{} & \multicolumn{3}{c|}{FPRF} & \multicolumn{3}{c|}{Ours} & \multicolumn{3}{c|}{FPRF} & \multicolumn{3}{c|}{Ours} & \multicolumn{3}{c|}{FPRF} & \multicolumn{3}{c|}{Ours} & \multicolumn{3}{c|}{FPRF} & \multicolumn{3}{c}{Ours} \\
\midrule
\multicolumn{1}{l|}{$d$} &    t &   FP &   FN &    t &    FP & FN &    t &   FP &   FN &     t &    FP & FN &     t &      FP &     FN &    t &       FP & FN &     t &      FP &    FN &    t &       FP & FN \\
\midrule
\multicolumn{1}{l|}{$10^{-2}$}   & 2.41 & 1.8K & 4 & 18.86K & 2.6K & 0 & 1.16 & 3.3K & 19 & 9.64K & 4.8K & 0 & 8.04 & 869.1K & 1 & 12.04 & 8.1M & 0 & 8.01 & 1.1M & 0 & 19.12 & 8.3M & 0\\
\multicolumn{1}{l|}{$10^{-8}$}   & 4.53 & 83 & 3 & 1.60K & 159 & 0 & 0.60 & 160 & 28 & 3.42K & 309 & 0 & 8.00 & 4 & 2 & 0.72 & 8 & 0 & 0.77 & 16 & 0 & 0.73 & 40 & 0\\
\multicolumn{1}{l|}{$10^{-16}$}    & 2.23 & 29 & 69 & 1.51K & 108 & 0 & 0.55 & 45 & 145 & 2.92K & 214 & 0 & 7.78 & 0 & 5.2K & 0.71 & 2 & 0 & 0.25 & 0 & 2.3K & 0.72 & 17 & 0\\

\multicolumn{1}{l|}{$10^{-30}$}    & 2.24 & 9 & 70 & 1.39K & 108 & 0 & 0.58 & 5 & 147 & 2.79K & 214 & 0 & 7.77 & 0 & 5.2K & 0.66 & 2 & 0 & 0.25 & 0 & 2.3K & 0.67 & 17 & 0\\

\multicolumn{1}{l|}{$10^{-100}$}    & 2.31 & 9 & 70 & 1.43K & 108 & 0 & 0.80 & 5 & 147 & 2.82K & 214 & 0 & 7.75 & 0 & 5.2K & 0.66 & 2 & 0 & 0.25 & 0 & 2.3K & 0.68 & 17 & 0\\

\bottomrule
\end{tabular}
    \label{tab:us-vs-fprf}
\end{table*}